\documentclass[12pt]{article}
\usepackage[margin=1.0in]{geometry}
\usepackage[T1]{fontenc}
\usepackage[font=small, labelfont=bf]{caption}
\usepackage{graphicx}
\usepackage{amsmath}
\usepackage{heck}
\usepackage{amssymb}
\usepackage{amsthm}
\usepackage{longtable}
\usepackage{slashed}
\usepackage{verbatim}
\usepackage{cite}
\usepackage{enumerate}
\usepackage{subfig}  
\usepackage{soul}
\usepackage{url}
\usepackage{hyperref} 
\usepackage[usenames,dviptsnames]{xcolor}
\usepackage{mdframed}

\usepackage{tikz}
\usepackage{tikz-cd}

\newcommand{\Tr}{{\rm Tr}}

\newcommand{\rmd}{\textrm{d}}
\newcommand{\JE}{J^{\textrm{E}}}
\newcommand{\JM}{J^{\textrm{M}}}

\newcommand{\Rend}{R_{\text{endable}}}

\theoremstyle{definition}
\newtheorem{thm}{Theorem}

\newtheorem{lemma}[thm]{Lemma}

\newcommand{\beq}{\begin{eqnarray}}
\newcommand{\eeq}{\end{eqnarray}}

\usepackage{titlesec}

\DeclareMathOperator{\Inn}{Inn}
\DeclareMathOperator{\Aut}{Aut}
\DeclareMathOperator{\Out}{Out}

\mdfdefinestyle{statementstyle}{
   frametitlerule = true,
   frametitlefont = {\normalfont\bfseries},
   frametitlebackgroundcolor = gray!10,
}
\mdtheorem[style=statementstyle]{statement}{Statement}

\numberwithin{equation}{section}

\date{\today}
\title{
Non-Invertible Global Symmetries \\ and Completeness of the Spectrum
}
\institution{AMHERST}{\centerline{${}^{1}$Department of Physics, University of Massachusetts, Amherst, MA 01003 USA}}
\institution{HARVARD}{\centerline{${}^{2}$Department of Physics, Harvard University, Cambridge, MA 02138 USA}}
\institution{BERKELEY}{\centerline{${}^{3}$Department of Physics, University of California, Berkeley, CA 94720, USA}}
\institution{PRINCETON}{\centerline{${}^{4}$School of Natural Sciences, Institute for Advanced Study, Princeton, NJ 08540, USA}}
\authors{Ben Heidenreich\worksat{\AMHERST}\footnote{e-mail: {\tt bheidenreich@umass.edu}}, Jacob McNamara\worksat{\HARVARD}\footnote{e-mail: {\tt jmcnamara@g.harvard.edu}}, Miguel Montero\worksat{\HARVARD}\footnote{e-mail: {\tt mig.montero.m@gmail.com}}, Matthew Reece\worksat{\HARVARD}\footnote{e-mail: {\tt mreece@g.harvard.edu}}, Tom Rudelius\worksat{\BERKELEY,\PRINCETON}\footnote{e-mail: {\tt rudelius@berkeley.edu}}, and Irene Valenzuela\worksat{\HARVARD}\footnote{e-mail: {\tt ivalenzuela@g.harvard.edu}}}

\begin{document}
\hfill ACFI-T21-03 \vspace{-1.5\baselineskip} 
\abstract{
It is widely believed that consistent theories of quantum gravity satisfy two basic kinematic constraints: they are free from any global symmetry, and they contain a complete spectrum of gauge charges. For compact, abelian gauge groups, completeness follows from the absence of a 1-form global symmetry. However, this correspondence breaks down for more general gauge groups, where the breaking of the 1-form symmetry is insufficient to guarantee a complete spectrum. We show that the correspondence may be restored by broadening our notion of symmetry to include non-invertible topological operators, and prove that their absence is sufficient to guarantee a complete spectrum for any compact, possibly disconnected gauge group. In addition, we prove an analogous statement regarding the completeness of \emph{twist vortices}: codimension-2 objects defined by a discrete holonomy around their worldvolume, such as cosmic strings in four dimensions. We discuss how this correspondence is modified in various, more general contexts, including non-compact gauge groups, Higgsing of gauge theories, and the addition of Chern-Simons terms. Finally, we discuss the implications of our results for the Swampland program, as well as the phenomenological implications of the existence of twist strings.
}
\maketitle

\setcounter{tocdepth}{2}
\tableofcontents

\section{Introduction}
\label{sec:intro}

Although the landscape of quantum gravity theories may be vast, certain features seem to be universally true of all such theories. One such feature is the absence of global symmetries, including $p$-form global symmetries, for which the charged operators are supported on manifolds of dimension $p$ \cite{Gaiotto:2014kfa}. Another such feature is completeness of the spectrum---the presence of particles (or multiparticle states) transforming in every representation of the gauge group \cite{polchinski:2003bq}. Compelling evidence for the absence of global symmetries in quantum gravity has been given in \cite{banks:1988yz, kallosh:1995hi, Banks:2010zn, Harlow:2018tng, Harlow:2020bee, Chen:2020ojn,Belin:2020jxr, Yonekura:2020ino}, while arguments for completeness of the spectrum---also known as the ``Completeness Hypothesis''---were provided in \cite{Banks:2010zn, Harlow:2018tng}.

It has often been remarked that the absence of global symmetries and completeness of the spectrum are related to one another. Although the motivation for these two conjectures comes from gravity, their relationship is a purely field theory statement that can be studied in the context of effective quantum field theories, without committing to a particular UV completion.\footnote{Statements similar to ours have been discussed in the framework of algebraic QFT in~\cite{Casini:2020rgj}. In some cases, the underlying assumptions made there are significantly stronger than ours. Our arguments apply quite generally, for example, to EFTs (independent of their UV completion or lack thereof) or to TQFTs.} In what follows, we will make this relationship precise in the pure gauge theory of a connected, compact gauge group $G$: such a theory has a 1-form ``electric'' global symmetry, with symmetry group given by the center $Z(G)$, whose charged operators are Wilson lines. Such a symmetry is characterized by the presence of topological codimension-2 operators $U_g$, each labeled by an element $g \in Z(G)$, which fuse according to the group multiplication law:
\begin{equation}
U_g \times U_{g'} = U_{g''} \,,~~~~g'' = g g'.
\label{introfuse}
\end{equation}
This 1-form symmetry is explicitly broken to a subgroup in the presence of charged matter, and we show that it is broken completely if and only if the spectrum is complete. Thus, in such a theory, absence of the 1-form electric symmetry is in 1-1 correspondence with completeness of the spectrum. Note that this statement applies to both quantum field theories as well as quantum gravities.

However, this correspondence between the absence of global symmetries and completeness does not hold in general, as pointed out in reference~\cite{Harlow:2018tng}: 
a finite, nonabelian gauge group $G$, such as $S_4$, may have a trivial center, so it does not have a 1-form electric symmetry even if its spectrum is incomplete. Nonetheless, there is a generalization of the no global symmetries--completeness correspondence that applies to finite gauge groups \cite{Rudelius:2020orz}: completeness of the spectrum is equivalent to the absence of certain codimension-2 topological operators, known as Gukov-Witten operators \cite{Gukov_2006,Gukov_2010}, which are labeled by conjugacy classes of $G$. If $G$ is abelian, these Gukov-Witten operators generate the 1-form electric global symmetry. But if $G$ is nonabelian, they satisfy a more complicated fusion algebra than the one of \eqref{introfuse}, and in particular not every topological Gukov-Witten operator will have an inverse. In this sense, we might say that completeness of the spectrum is in 1-1 correspondence with the absence of \emph{non-invertible} 1-form electric global symmetries, which are characterized by the presence of topological, non-invertible codimension-2 operators.

In this paper, we will see that a similar story applies to \emph{all} compact gauge groups: completeness of the gauge theory spectrum is equivalent to the absence of (possibly non-invertible) 1-form electric global symmetries. When $G$ is connected, finite and abelian, or a simple direct product of a connected group with a finite abelian group, the electric 1-form symmetry is an ordinary, invertible symmetry. But when $G$ is finite and nonabelian, a nontrivial semidirect product of an abelian group with a connected group, or a product of a nonabelian finite group with a connected group, the electric 1-form symmetry will be non-invertible.

One example of the latter case is $O(2) \simeq U(1) \rtimes \mathbb{Z}_2$ gauge theory, which can be thought of as $U(1)$ gauge theory with the $\mathbb{Z}_2$ charge conjugation symmetry gauged. We will study this theory in detail, and we will see that the pure gauge theory has a non-invertible 1-form symmetry with a continuous family of non-invertible topological codimension-2 operators. In the presence of charged matter, this continuous non-invertible global symmetry may be broken to a discrete, non-invertible ``subgroup.'' If the spectrum is complete, it is broken entirely, and there are no topological Gukov-Witten operators whatsoever. For general compact gauge groups, we prove the following statement:
\begin{statement}[Electric Completeness vs.~Topological Gukov-Witten Operators]
Consider a gauge theory with compact gauge group $G$ coupled to a set of matter fields transforming in representations of $G$. Then the theory is electrically complete (i.e., states exist transforming in all possible representations of $G$) if and only if there are no topological Gukov-Witten operators in the theory.
\label{stat:ElectricCompleteness}
\end{statement}
This correspondence breaks down in the case of noncompact gauge groups, however: $\mathbb{R}$ gauge theory may have an incomplete (albeit dense) spectrum of states without any topological Gukov-Witten operators.

One objection to discussing completeness for nonabelian gauge groups is the possibility of confinement. In a confining gauge theory, all asymptotic particle states in flat space are neutral under the gauge symmetry, and so the sense in which charged states exist must be clarified. Relatedly, the statement of completeness relies on ``the'' gauge group, which is not a duality-invariant notion. In reference \cite{Harlow:2018tng}, this was resolved by restricting to completeness for ``long-range gauge symmetry,'' a distinct physical concept that is manifestly duality invariant, but which excludes confining gauge theories. Since we would like to discuss confining gauge theories, we need a different definition.

Our definition of completeness is that a charged state ``exists'' if the corresponding Wilson line operator may end on a charged point operator, which we think of as creating the charged state (see \cite{Rudelius:2020orz} for related discussion). Equivalently, the charged state exists in the defect Hilbert space on $S^{d - 1}$ with an insertion of a defect Wilson line at a point. In a confining theory, the energy of such a state will diverge as the volume of $S^{d - 1}$ grows, but the existence of such a state on a finite volume sphere shows that the theory kinematically includes charged states, even if they are dynamically confined in the IR. This definition may be made duality-invariant by requiring such states not only for Wilson lines, but for every line operator in the theory, a notion referred to as \emph{total completeness} \cite{Rudelius:2020orz}.

A similar story plays out on the magnetic side of things: continuous gauge groups feature 't Hooft operators of dimension $d-3$, whose topological classes are labeled by elements of $\pi_1(G)/\pi_0(G)$, where an element of $\pi_0(G)$ acts on a path in $\pi_1(G)$ via conjugation.\footnote{Note that $\pi_0(G)$ is a group, with the group structure descending from the multiplication in $G$.} If the action of $\pi_0(G)$ is trivial, then there is a magnetic $(d-3)$-form symmetry with group $\pi_1(G)^\vee$, the Pontryagin dual group of $\pi_1(G)$. This symmetry group is broken in the presence of dynamical, magnetically charged objects of dimension $d-3$ (monopoles, in four dimensions), and it is broken completely if and only if the magnetic spectrum is complete. If $\pi_0(G)$ acts nontrivially on $\pi_1(G)$, however---as in the case of $O(2)$ gauge theory---the magnetic $(d-3)$-form symmetry will be non-invertible.

In addition, whenever $\pi_0(G)$ is nontrivial, there will be a $(d-2)$-form, possibly non-invertible global symmetry generated by topological Wilson lines, which are in 1-1 correspondence with representations of $\pi_0(G)$. In this paper, we will not prove any general statement regarding the magnetic non-invertible $(d-3)$-form symmetry, but instead focus primarily on the latter $(d-2)$-form symmetry generated by the Wilson lines. The charged operators under this $(d-2)$-form symmetry are the Gukov-Witten operators themselves, and the symmetry will be broken in the presence of certain dynamical $(d-2)$-dimensional objects, which we will refer to as ``twist vortices.'' The global symmetry is broken entirely if and only if the spectrum of twist vortices is complete. Again, we prove a general statement:
\begin{statement}[Twist Vortex Completeness vs.~Topological Wilson Operators]
Consider a gauge theory with compact gauge group $G$ coupled to a set of dynamical twist vortices, which give rise to a holonomy when a charged probe particle circles them. Then the spectrum of twist vortices is complete (i.e., multi-vortex states of all allowed charges exist) if and only if there are no topological Wilson line operators in the theory.\label{stat:TwistVortexCompleteness}\end{statement}

In quantum gravity, it is natural to suspect that all such global symmetries---including non-invertible global symmetries---are broken. 1-form electric global symmetries may be broken simply by adding weakly-coupled charged matter in every representation of the gauge group, while the $(d-2)$-form and $(d-3)$-form magnetic global symmetries require the existence of magnetically-charged objects of appropriate dimension. One way to produce these objects within Lagrangian effective field theory is to un-Higgs the gauge group $G$ to a larger gauge group $\tilde{G}$. If $\tilde{G}$ is connected, then $\pi_0(\tilde{G})$ is trivial, and the $(d-2)$-form (possibly non-invertible) global symmetry will be broken completely. If $\tilde{G}$ is simply connected, then $\pi_1(\tilde{G})$ is also trivial, and the $(d-3)$-form magnetic global symmetry will also be broken completely. 

The relationship between completeness and the absence of 1-form electric global symmetries is modified in the presence of nontrivial Chern-Simons terms. These terms mix various electric and magnetic higher-form global symmetries into a single structure known as a ``higher-group'' global symmetry \cite{Cordova:2018cvg, Hidaka:2020iaz, Hidaka:2020izy}. In simple examples, such as $U(1)$ gauge theory with a BF-coupling or a $\theta F \wedge F$ coupling, we will see that this mixing can destroy the simple relationship between completeness (or magnetic completeness) of the spectrum and the absence of electric $1$-form global symmetries (or magnetic $(d-3)$-form global symmetries). However, the relationship between the endability of all extended operators and the absence of any topological operator persists in the presence of these Chern-Simons terms.

The remainder of the paper is structured as follows: in Section \ref{sec:Top}, we review the notion of a $p$-form global symmetry, topological operators, and the linking and fusion of such operators in pure gauge theories. In Section \ref{sec:Connected}, we establish a relationship between completeness/magnetic completeness and the absence of ordinary, invertible electric/magnetic global symmetries for compact, connected gauge groups, illustrating our results in $U(1)$ and $SU(N)/\mathbb{Z}_K$ gauge theory. In Section \ref{sec:Discrete}, we review the relationship between completeness and the absence of \emph{non-invertible} global symmetries for finite gauge groups, illustrating our results in $S_3$ gauge theory. In Section \ref{sec:Disconnected}, we extend this result to general compact, disconnected gauge groups, and we examine the non-invertible global symmetries of $O(2)$ gauge theory in depth. In Section \ref{sec:Noncompact}, we consider noncompact examples of $\mathbb{R}$ gauge theory and $\mathbb{Z}$ gauge theory, showing that the relationship between completeness and the absence of global symmetries breaks down in these cases. In Section \ref{sec:Higgsing}, we examine how (non-invertible) electric and magnetic global symmetries behave under Higgsing, and we show how these symmetries can be broken by un-Higgsing to a larger gauge group. In Section \ref{sec:CS}, we study how Chern-Simons terms affect our results. Finally, we end in Section \ref{sec:CONC} with conclusions and remarks on how our results fit into the larger Swampland program, as well as the phenomenological implications of twist vortices.

\section{Topological Operators in Gauge Theories}\label{sec:Top}

In this section, we review aspects of topological operators and higher-form symmetries, both in the general setting and in the special case of gauge theories. In particular, we review the notions of Wilson operators and Gukov-Witten operators in gauge theories, and derive the conditions under which these operators are topological in pure gauge theory.

\subsection{Higher-form Symmetries and Topological Operators}
\label{sec:generalTop}

As discussed in \cite{Gaiotto:2014kfa}, a $p$-form global symmetry is characterized by the presence of ``charge operators'' $U_g(\mathcal{M}^{(d-p-1)})$, also known as symmetry generators, each of which carries support on a closed manifold $\mathcal{M}^{(d-p-1)}$ of codimension $p+1$ and is labeled by an element $g$ of the symmetry group. These operators fuse according to the group multiplication law, namely,
\begin{equation}
U_g(\mathcal{M}^{(d-p-1)}) \times U_{g'}(\mathcal{M}^{(d-p-1)}) = U_{g''}(\mathcal{M}^{(d-p-1)}),~~~~g'' = g g'.
\label{eq:fusion}
\end{equation}
In particular, every symmetry generator $U_g$ has an inverse, given by $U_g^{-1} = U_{g^{-1}}$, such that $U_g \times U_{g^{-1}} =1$, the identity operator. Thus, we say that these symmetry generators are \emph{invertible}. They are also \emph{topological} in the sense that small, continuous deformations of $\mathcal{M}^{(d-p-1)}$ do not affect any physical observables provided $\mathcal{M}^{(d-p-1)}$ does not cross any charged operators in the deformation process.

Symmetries are associated with Ward identities. In the language of these symmetry generators, the Ward identity says that a symmetry generator $U_g(S^{(d-p-1)})$ supported on a sphere surrounding a charged operator $V(\mathcal{C}^{(p)})$ in some representation $\rho_V$ of the global symmetry is equal to the charged operator itself multiplied by $\rho_V(g)$, the action of $g$ in the representation $\rho_V$:
\begin{equation}
U_g(S^{d-p-1}) V(\mathcal{C}^{(p)}) =  \rho_V(g) V(\mathcal{C}^{(p)}) .
\label{eq:Ward}
\end{equation}
When the global symmetry group is abelian, which is necessarily the case when $p \geq 1$, $\rho_V(g)$ is simply a phase. The Ward identity can be understood as the result of shrinking the sphere $S^{d-p-1}$ to zero size on the operator $V$; such a shrinking is allowed due to the topological nature of the operator.

Not all topological operators are associated with global symmetries of the type considered above. Rather than satisfying a fusion algebra of the form \eqref{eq:fusion}, a more general topological operator of dimension $d-p-1$ may satisfy a fusion algebra of the form
\begin{equation}
T_a(\mathcal{M}^{(d-p-1)})  \times T_b(\mathcal{M}^{(d-p-1)})  = \sum_c N^c_{ab} T_c(\mathcal{M}^{(d-p-1)}) \,,
\label{fusion_category_fusion}
\end{equation}
for some integer coefficients $N^c_{ab}$. In particular, if $T_a$ and $T_b$ are symmetry generators, the sum on the right-hand side is given by a single term with coefficient 1, but for a more general topological operator this is not the case. Relatedly, a general topological operator $T_a$ is not necessarily invertible: there may not exist any topological operator $T_b$ such that $T_a \times T_b$ is equal to the identity. In this paper, we will refer to a set of such non-invertible topological operators which is closed under fusion, along with their associated fusion algebra, as a \emph{non-invertible global symmetry}. This notion has appeared before in the context of 2-dimensional CFTs (see, e.g., \cite{Verlinde:1988sn, Moore:1988qv, Frohlich:2006ch, Davydov:2010rm, Bhardwaj:2017xup, Chang:2018iay}) as well as more general condensed matter systems (see, e.g., \cite{Ji:2019jhk, Kong:2020jne, Kong_2020,Komargodski:2020mxz}) under various other names, including \emph{algebraic higher symmetry}, \emph{categorical symmetry}, and \emph{fusion category symmetry}, referring to the more complicated fusion algebra \eqref{fusion_category_fusion}. 

These more general topological operators may link with a ``charged'' operator $V(\mathcal{C}^{(p)})$ of dimension $p$, generalizing the Ward identity of \eqref{eq:Ward}:
\begin{equation}
T_a(S^{d-p-1}) V(\mathcal{C}^{(p)}) = B_V(a) V(\mathcal{C}^{(p)}) .
\label{eq:linking}
\end{equation}
Here, $B_V(a)$ is a linking coefficient.\footnote{In principle, one might imagine that $B_V(a)$ could be a matrix, but in every case we study it is simply a number.} We say that the operators $T_a$ and $V$ {\em link nontrivially} when $B_V(a) \neq B_1(a)$ where $B_1(a)$ is the linking coefficient with the identity operator (in particular, $B_V(a) = 0$ is nontrivial linking!). Like the equations above, this should be understood as an operator equation, valid within general correlation functions provided that there are no other operator insertions that link nontrivially with $S^{d-p-1}$ or $\mathcal{C}^{(p)}$. As before, this can be understood as the result of shrinking the sphere $S^{d-p-1}$ to zero size on the operator $V$. Note that \eqref{eq:linking} will lead to constraints on correlation functions of charged operators, similar to the case of an ordinary symmetry, although the non-invertibility and more complicated fusion algebra \eqref{fusion_category_fusion} will lead to qualitative differences.

We will refer to $B_1(a)$ as the \emph{quantum dimension} of the operator $T_a$, $\text{dim}(T_a)$.\footnote{For operators of higher dimension, there are multiple notions of quantum dimension one could construct, corresponding to different topologies of the support of the operator, and which do not obviously agree. We thank Meng Cheng, Ryan Thorngren, Xiao-Gang Wen, and Xueda Wen for discussions on this point.} This nomenclature originates from the literature on nonabelian anyons in $(2+1)$d TQFTs (see, e.g., the review \cite{Nayak:2008zza}). Anyons are inserted by line operators and, because a system with $n$ anyons is protected by a gap, the dimension of the Hilbert space $d_n$ is well defined. The quantum dimension is the asymptotic value of $d_n/n$ as $n\rightarrow\infty$, and it is not necessarily an integer.  In this paper, we consider an analogous concept in $d$ dimensions.

Quantum dimensions multiply under fusion and sum under addition of surfaces:
\begin{align}
\text{dim}(T_a \times T_b) &= \text{dim}(T_a) \times \text{dim}(T_b) \label{eq:qdimfuse}\\
\text{dim}(T_a + T_b) &= \text{dim}(T_a) + \text{dim}(T_b).
\end{align}
Note that the trivial surface operator has quantum dimension 1, and in fact so does any invertible operator, including charge operators for $p$-form symmetries.

So far, we have been considering the charged operator $V(\mathcal{C}^{(p)})$ defined on a manifold $\mathcal{C}^{(p)}$ without boundary. Suppose now, however, that $V(\mathcal{C}^{(p)})$ may be defined on a manifold $\mathcal{C}^{(p)}$ with boundary. This is not always possible in a given theory; if it is, we say that the operator $V$ is \emph{endable}. If a topological operator $T_a(S^{d-p-1})$ surrounds an endable operator $V(\mathcal{C}^{(p)})$ supported on some manifold $\mathcal{C}^{(p)}$ with boundary, $T_a(S^{d-p-1})$ may be either (a) shrunk to a point, yielding a factor of $B_V(a)$; or (b) unlinked from $V$ and then shrunk to a point (see Figure \ref{fig:toplinking}), yielding a factor of $B_1(a)$. This implies $B_V(a) = B_1(a)$, or in other words, any endable operator must link trivially with any topological operator. If $T_a = U_g$ is an invertible operator, we learn that any endable operator of dimension $p$ cannot carry charge under a $p$-form global symmetry. This relationship between endability of $V$ and the topological nature of $T_a$ will show up repeatedly in what follows.

\begin{figure}
\centering
\includegraphics[width=90mm]{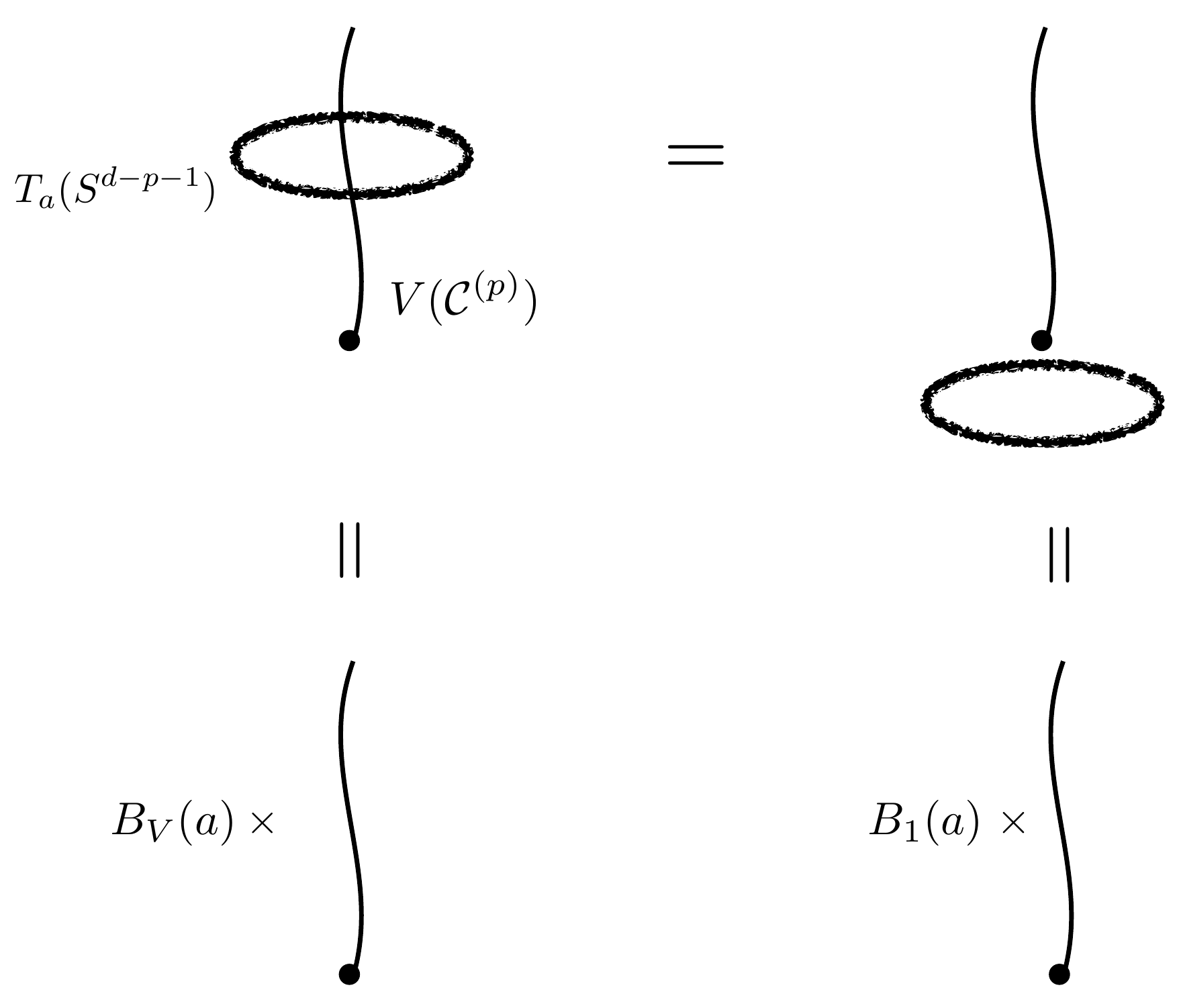}
\caption{Endable operators and topological operators. If a dimension $p$ charged operator $V(\mathcal{C}^{(p)})$ can end, a topological operator $T_{a}(S^{d-p-1})$ surrounding it may be shrunk to a point, yielding a factor of $B_V(a)$ (left), or it may be unlinked, yielding a trivial factor $B_1(a)$. This means that if $T_a$ does not satisfy $B_V(a)=B_1(a)$, then it cannot be topological whenever $V(\mathcal{C}^{(p)})$ is endable.}
\label{fig:toplinking}
\end{figure}

Finally, note that the property of being endable is closed under fusion: given two endable operators $V_1(\mathcal{C}^{(p)})$, $V_2(\mathcal{C}^{(p)})$ supported on a manifold with boundary $\mathcal{C}^{(p)}$, we may fuse them to form a new endable operator $V_3(\mathcal{C}^{(p)})$. Similarly, the property of being topological is also closed under fusion: if two surfaces $T_a$ and $T_b$ are topological, their fusion $T_a \times T_b$ will be as well.

\subsection{Wilson Operators and Gukov-Witten Operators}
\label{sec:WandGW}

We will illustrate the above abstract discussion with examples arising in gauge theories: Wilson operators and Gukov-Witten operators. By gauge theory, we mean a quantum field theory defined by a path integral over connections for some Lie group $G$,\footnote{For us, Lie groups include discrete groups, as they are Lie groups of dimension zero. For discrete groups, the Yang-Mills action is trivial, and the path integral is simply a count of $G$ bundles modulo gauge.} with the Yang-Mills action (in particular, some of our arguments must be modified for Chern-Simons theories). Wilson and Gukov-Witten operators may or may not be endable, and they may or may not be topological. In this subsection, we will review the definitions of these operators. In the remainder of this section, we will characterize the conditions under which they are endable and topological in pure gauge theories.

In a gauge theory, there is a natural set of line operators, the \emph{Wilson lines}, defined by the path-ordered exponential 
\begin{equation} 
W_\rho(\gamma) = \text{Tr}_{\rho}{\rm P} \left[ \exp\left(i \oint_{\mathcal{C}} A\right) \right],
\end{equation}
where $\rho$ is a representation of $G$ and $\mathcal{C}$ is an arbitrary closed curve. These lines have a natural physical interpretation as the worldvolume of a very massive probe (i.e., non-dynamical) particle charged under $G$ in the representation $\rho$. 

Wilson line operators may or may not be topological. For example, all of them are topological if $G$ is discrete, while none of them are topological if $G$ is compact and connected. A Wilson line is called \emph{endable} if it may be defined on an open curve with endpoints. A Wilson line in representation $\rho$ can end in a local charged operator in the same representation.\footnote{Such a charged operator is only well-defined when a Wilson line is attached, so strictly speaking it is not a genuine local operator at all. When we say a local charged operator, we are using a common abuse of terminology that simply indicates that the end of the Wilson line is pointlike.} This operator creates a charged particle (or collection of particles), so a Wilson line can end if there are states (particles) in the gauge theory transforming in the same representation $\rho$.

Two Wilson lines supported on the same curve $\gamma$ fuse according to the tensor product operation:
\begin{align}
\rho \otimes \mu = \bigoplus_{i} \nu_i ~~\Leftrightarrow ~~ W_\rho \times W_\mu = \sum_i W_{\nu_i}.
\label{Wilsonfusion}
\end{align}
In particular, Wilson line operators with $\dim(\rho) \neq 1$ are not invertible. Physically, this fusion can be understood by thinking of the Wilson line $W_\rho$ as a probe particle in the representation $\rho$. A pair of such probe particles in the representations $\rho$ and $\mu$ together form a multiparticle state in the representation $\rho \otimes \mu$, which by \eqref{Wilsonfusion} can be decomposed into a sum over states in the representations $\nu_i$.

Strictly speaking, \eqref{Wilsonfusion} should only be taken literally for topological Wilson lines, giving a well-defined fusion algebra in the sense of Section \ref{sec:generalTop}. By contrast, non-topological Wilson lines are not closed under fusion; there are short-distance singularities, and additional operators may appear when fusing two Wilson lines.\footnote{We thank Shu-Heng Shao for discussions on this point.} In spite of these subtleties, the OPE of two non-topological Wilson lines labeled by $\rho$ and $\mu$ must include Wilson lines labeled by every representation in the right-hand side of \eqref{Wilsonfusion}, as the Wilson lines labeled by $\{\nu_i\}$ taken together insert a complete set of states in the full reducible representation $\rho\otimes\mu$.\footnote{In general, the endability of a direct sum of operators only implies the endability of at least one of the summands. However, in the case of fusion \eqref{Wilsonfusion}, this complete set of states ensures that each summand $W_{\nu_i}$ is endable if each of $W_\rho$ and $W_\mu$ are endable.} Because of this, we will refer to \eqref{Wilsonfusion} as the fusion structure even for non-topological Wilson lines, as our interest is in the completeness of the charge spectrum, a kinematic question that isn't affected by the non-universal features of the Wilson line OPE.

Finally, another way to define a codimension-$k$ operator in a gauge theory is to excise some $k$-dimensional locus from spacetime and specify boundary conditions for the gauge field on the rest of the geometry. For instance, one could excise a codimension-2 manifold $\mathcal{M}^{(d-2)}$, and specify a choice of connection on the transverse $S^1$. $G$-connections on $S^1$ modulo gauge equivalence are classified by elements of $G$ (which specify the holonomy around the circle) modulo the conjugacy action of $G$ on itself (which implements gauge transformations). This is the same as the set of conjugacy classes of $G$. The resulting set of codimension-2 operators $T^{\text{GW}}_a(\mathcal{M}^{(d-2)})$, where $a$ is a conjugacy class in $G$,  are called \emph{Gukov-Witten operators} \cite{Gukov_2006, Gukov_2010}.  Just like Wilson lines, Gukov-Witten operators admit a simple physical interpretation: they correspond to insertions of probe (non-dynamical) \emph{vortices}, codimension-2 objects defined by the nontrivial gauge holonomy around their worldvolume.\footnote{For example, in $d = 4$ these are strings, while in $d = 10$ they are $7$-branes, such as D7-branes in Type IIB string theory.} Outside a vortex, the gauge field is locally pure gauge, but globally this is not the case.

For completeness, we describe the fusion of Gukov-Witten operators supported on the same codimension-2 manifold, following reference \cite{Dijkgraaf:1989hb}. Suppose we want to compute the fusion of Gukov-Witten operators labeled by two conjugacy classes $a$ and $b$. If we have two representatives $g \in a$, $g' \in b$, then the conjugacy class $c = [gg']$ of the product is an allowed fusion channel for the Gukov-Witten operators. Summing over all representatives, we obtain the fusion rule
\begin{equation}\label{GW_operator_fusion}
T^{\text{GW}}_a \times T^{\text{GW}}_b = \sum_{c} N_{ab}^c T^{\text{GW}}_{c},
\end{equation}
where the fusion coefficient $N_{ab}^c$ is the number of $G$-orbits in the set,
\beq
S = \left\{ (g, g', g'') \in a \times b \times c\ |\ g g' = g'' \right\},
\eeq
of ways of multiplying elements of $a$ and $b$ to obtain an element of $c$. As with \eqref{Wilsonfusion}, for non-topological Gukov-Witten operators we must only interpret \eqref{GW_operator_fusion} as describing a universal part of the Gukov-Witten operator OPE. We show in Section \ref{sec:topGW} that topological Gukov-Witten operators always correspond to conjugacy classes with finitely many elements, so the fusion coefficients $N_{ab}^c$ are finite integers when the conjugacy classes $a, b, c$ correspond to topological Gukov-Witten operators.

\subsection{Topological Operators in Pure Gauge Theories}

In this subsection, we will explain how to characterize which Wilson lines and Gukov-Witten operators are topological in pure gauge theory, i.e., in the absence of charged particles or twist vortices. We will see that this is correlated with which Gukov-Witten and Wilson lines operators (respectively) are endable. In the rest of the paper, starting in Section \ref{sec:Connected}, we will explain how this story changes in the presence of dynamical charged states. We also recommend that readers less familiar with the subject first take a look at the examples in Sections \ref{sec:Connected}, \ref{sec:Discrete}, and \ref{sec:Disconnected} if the following discussion is too abstract.

\subsubsection{Topological Gukov-Witten Operators and Endable Wilson Lines}
\label{sec:topGW}

Gukov-Witten operators are sometimes, but not always, topological. A familiar case in which they are topological is in free $U(1)$ gauge theory, where the Gukov-Witten operators generate a 1-form electric symmetry. On the other hand, in pure $SU(2)$ gauge theory, the only nontrivial topological Gukov-Witten operator is the one generating the $\mathbb{Z}_2$ 1-form symmetry; the rest are not topological. In this subsection, we ask: under what conditions is a Gukov-Witten operator in a pure gauge theory topological? The answer can be understood from the endability of the Wilson lines with which Gukov-Witten operators link.

Let us begin by assuming that a Gukov-Witten operator $T^{\text{GW}}_a(S^{d-2})$ is topological, and understanding how to specialize the general linking equation~\eqref{eq:linking} to this context. In the Euclidean picture, the Aharonov-Bohm phase is encoded in the linking properties of the corresponding Wilson line $W_\rho(\gamma)$ and the Gukov-Witten operator $T^{\text{GW}}_a(\mathcal{M}^{(d-2)})$. Specifically, if $S^{d-2}$ is a $(d-2)$-sphere that links the curve $\gamma$, then 
\begin{equation}
 T^{\text{GW}}_a(S^{d-2}) W_\rho(\gamma)= \frac{\chi_{\rho}(a)}{\chi_{\rho}(1)} \text{size}(a) \,  W_\rho(\gamma),
 \label{GWlink}
 \end{equation}
where we have defined the \emph{character} of the conjugacy class $a$ in representation $\rho$,
\begin{equation} \chi_{\rho}(a):= \text{Tr}(\rho(g_1)).\label{ABph1}\end{equation}

\begin{figure}
\begin{center}
\includegraphics[width=78mm]{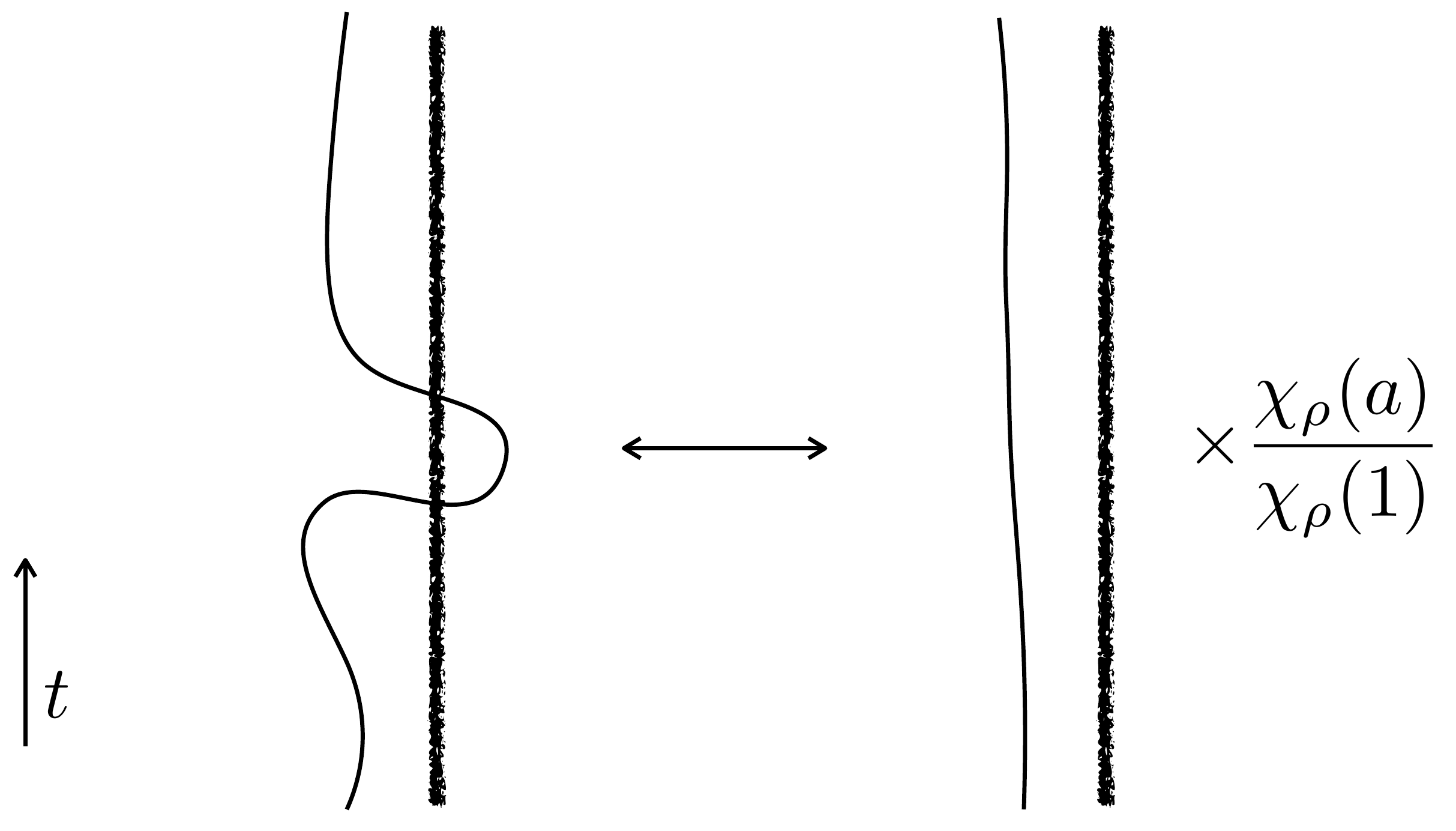}
\caption{Cross-section of a particle worldline going around a vortex; the time direction is shown, and the $(d - 3)$ spatial directions along which the vortex extends are suppressed. The fact that the particle winds around the vortex means that the corresponding Wilson line and Gukov-Witten operators are linked, in the Euclidean picture. The linking property \eqref{eq:linking} of topological and charged operators encodes the fact that the correlator on the left and right pictures differ by a multiplication by the character of the representation, which is the Aharonov-Bohm factor.}
\label{f1GW}
\end{center}
\end{figure}

In the Lorentzian picture, this equation can be understood as a generalization of the Aharanov-Bohm effect: the holonomy carried by the probe vortex (i.e., the Gukov-Witten operator) can be measured via an Aharonov-Bohm experiment in which a probe particle in a representation $\rho$ (i.e., a Wilson line) is moved adiabatically around the vortex. This process is illustrated in Figure \ref{f1GW}.
By taking the probe particle to be in a fully mixed state in its gauge indices, one gets a gauge-invariant characterization of the holonomy around the vortex as \cite{Alford:1992yx}
\begin{equation}
\frac{1}{\text{dim}(\rho)} \text{Tr}(\rho(g_1)) := \frac{\chi_{\rho}(a)}{\chi_{\rho}(1)},
\label{nonab}
\end{equation}
which matches the first factor on the right-hand side of \eqref{GWlink}.\footnote{In the case of $U(1)$ gauge theory, and with a magnetic flux specified by $\theta\in U(1)$, the character $\chi_{q}(\theta)=\mathrm{e}^{iq\theta}$ is the usual Aharonov-Bohm phase; equation \eqref{nonab} is the natural nonabelian analog.} The second factor size$(a)$ is the quantum dimension of the Gukov-Witten operator, as defined in Section \ref{sec:generalTop}, associated with the process of shrinking the sphere $S^{d-2}$ in the absence of a Wilson line. This shrinking process produces a local operator that commutes with all other operators of the theory, so by Schur's lemma, it must be proportional to the identity. The constant of proportionality is given by number of elements in the conjugacy class $a$, i.e., $\text{dim}(T_a^{\rm GW}) := \text{size}(a)$. From this, we recover \eqref{eq:linking} with linking coefficient given by 
 \begin{equation} 
B_{\rho}(a)=\frac{\chi_{\rho}(a)}{\chi_{\rho}(1)}\, \text{size}(a)\,,
   \label{eq:linkingBrho}
\end{equation}
and we see that the linking of $T^{\text{GW}}_a(\mathcal{M}^{(d-2)})$ with $W_\rho(\gamma)$ is nontrivial when
\begin{equation}
\chi_\rho(a) \neq \chi_\rho(1) := \dim(\rho)\,.
\label{eq:nontriviality}
\end{equation}
A Gukov-Witten operator that is topological cannot link with an endable Wilson line, by the argument illustrated in Fig.~\ref{fig:toplinking} and discussed in Section~\ref{sec:generalTop}. Thus, if $\chi_\rho(a) \neq \chi_\rho(1)$ for an endable Wilson line $W_\rho$, we conclude that $T^{\text{GW}}_a$ cannot be topological.

To better understand how to interpret the condition~\eqref{eq:nontriviality} for nontrivial linking, it is useful to establish a lemma, which will be used several times in the remainder of the paper.
\begin{lemma}
Let $G$ be a compact Lie group, and $\rho$ a (complex) representation of $G$. Then, $\chi_{\rho}([g]) = \chi_\rho(1)$ if and only if $\rho(g)$ is the identity for all $g \in [g]$.
\label{lemma1}
\end{lemma}
 \begin{proof}
Clearly, if $\rho(g)$ is the identity for all $g \in [g]$, i.e., $\rho(g) = \rho(1) = I$, then $\chi_{\rho}([g]) =  \Tr_\rho(I) = \chi_\rho(1)$. Conversely, since $G$ is assumed compact, any complex representation can be chosen to be unitary. This means that the eigenvalues $\lambda_i, i = 1,...,$ dim$(\rho)$ of $\rho(g)$ must be roots of unity. Since $\chi_\rho(1) = \text{dim}(\rho)$, we see that $\sum_i \lambda_i = \chi_{\rho}([g])$ can be no larger than dim$(\rho)$, with equality if and only if $\lambda_i = 1$ for all $i$. Thus, $\rho(g) = I$.
 \end{proof}
 
 In pure $G$ gauge theory, the Wilson line in the adjoint representation is endable, as it may end on the field strength $F_{\mu \nu}$, which transforms in the adjoint. By considering multiple insertions of $F_{\mu \nu}$, we further obtain multiparticle states in tensor powers of the adjoint representation, and any Wilson line corresponding to a representation which appears in such a tensor power is thus endable. In fact, these are precisely the endable Wilson lines in pure $G$ gauge theory: 
 \begin{statement}[Endable Wilson Line Operators in Pure Gauge Theory] Consider a pure $G$ gauge theory. The endable Wilson line operators are precisely those corresponding to representations $\rho$ built from the adjoint under taking tensor products and sub-representations.
\label{stat:endableW}
\end{statement}
Hence, any Gukov-Witten operator that links nontrivially with an adjoint Wilson line cannot be topological. So, we have reduced the problem of classifying topological Gukov-Witten operators in pure gauge theory to that of finding conjugacy classes $[g]$ for which $\chi_{\text{adj}}([g]) \neq  \chi_{\text{adj}}(1)$.

By Lemma \ref{lemma1}, we see that the topological Gukov-Witten operators in pure $G$ gauge theory correspond precisely to conjugacy classes $[g]$ which act trivially on the adjoint representation.\footnote{More precisely, Lemma \ref{lemma1} establishes that a trivial action of the adjoint representation is a \emph{necessary} condition for a Gukov-Witten operator to be topological. It is well-known that this condition is \emph{sufficient} when $G$ is connected or finite: in Section \ref{ssec:Coulomb}, we will argue that this is sufficiently more generally by showing that the general case can be reduced to these special cases.} This is in turn equivalent to the statement that $g g_0 g^{-1} = g_0$ for all $g_0$ in the \emph{identity component} $G_0$, i.e., the connected component of $G$ containing the identity. (Note that if this holds for one representative $g \in [g]$, it holds for all such representatives.) The set of such $g$ is known as the {\em centralizer} of $G_0$ in $G$, denoted $Z_G(G_0)$.

To summarize:

\begin{statement}[Topological Gukov-Witten Operators in Pure Gauge Theory]
Consider pure $G$ gauge theory, with $G$ compact. The topological Gukov-Witten operators are precisely those corresponding to conjugacy classes contained in the centralizer $Z_G(G_0)$ of the identity component $G_0$ of the group $G$.
\label{stat:topGW}
\end{statement}

From what we have said so far, one might worry that the quantum dimension of a topological Gukov-Witten operator could be infinite, as the number of elements in a conjugacy class can be infinite. However, we now show that for a compact gauge group, the conjugacy classes corresponding to topological Gukov-Witten operators are of finite size, hence these surfaces have finite quantum dimension. In fact, the converse is true as well: every conjugacy class of finite size corresponds to a topological Gukov-Witten operator.

Suppose $T^{\text{GW}}_{a}(\mathcal{M}^{(d-2)})$ is topological. Fixing a representative $g$ of the conjugacy class $a = [g]$, we can write $g_1 = h_1 g h_1^{-1}$, $g_2 = h_2 g h_2^{-1}$ for any $g_1, g_2 \in a$. If $h_1$ and $h_2$ lie in the same connected component $G_1 \subset G$ then $h_1 h_2^{-1}$ lies in the identity component $G_0$. Since
\begin{equation}
g_1 = (h_1 h_2^{-1}) g_2 (h_1 h_2^{-1})^{-1}
\end{equation}
and $g_1, g_2 \in Z_G(G_0)$ per Statement~\ref{stat:topGW}, we then obtain $g_1 = g_2$. Thus, the number of distinct elements in $[g]$ is no greater than the number of connected components, which is finite for any compact group. In particular, this shows that the factor of $\text{size}(a)$ in \eqref{eq:linkingBrho} is finite. Conversely, if a conjugacy class contains finitely many elements, then it must be invariant under conjugation by any $h \in G_0$ by continuity. This implies that the conjugacy class is contained in $Z_G(G_0)$, and thus the corresponding Gukov-Witten operator is topological.

\subsubsection{Topological Wilson Lines and Endable Gukov-Witten Operators}
\label{sec:topW}

Just as we did for Gukov-Witten operators, we should ask which Wilson lines in a pure gauge theory are topological, as well as which Gukov-Witten operators are endable. In order to address both of these questions, we introduce the following lemma, which is a simple consequence of the Ambrose-Singer Theorem \cite{AmbroseSinger}.

\begin{lemma}
\label{ambrose_singer}
Let $A$ be a connection on a principal $G$ bundle over a manifold $X$, and suppose $\gamma : S^1 \to X$ is a contractible closed curve in $X$. Then the holonomy of $A$ around $\gamma$ is contained in the identity component $G_0$ of $G$.
\end{lemma}

\begin{proof}
This follows immediately from the Ambrose-Singer Theorem \cite{AmbroseSinger}, which states that the holonomy around a contractible closed curve is generated by the curvature, valued in the Lie algebra of $G$.
\end{proof}

With this lemma in hand, we immediately see that if the identity component $G_0$ acts trivially on a representation $\rho$, then $W_\rho(\gamma)$ is topological. In particular, let $\gamma_1$ and $\gamma_2$ be homotopic paths between a pair of spacetime points $p, q$. Then a closed curve $\gamma$ incorporating $\gamma_1$ can be deformed to a homotopic closed curve $\gamma'$ incorporating $\gamma_2$. This alters $W_\rho(\gamma)$ by inserting the holonomy around the contractible closed curve $\gamma_2^{-1} \gamma_1$ at the point $p$ on the Wilson line. Per Lemma~\ref{ambrose_singer}, this holonomy lies in $G_0$, hence it acts trivially on $\rho$ by assumption, and therefore $W_\rho(\gamma) = W_\rho(\gamma')$, i.e., $W_\rho(\gamma)$ is topological. Conversely, if the identity component $G_0$ acts non-trivially on $\rho$ then $W_\rho(\gamma)$ and $W_\rho(\gamma')$ will in general differ by insertions of the curvature, and therefore $W_\rho$ will not be topological. In summary:

\begin{statement}[Topological Wilson Line Operators in Pure Gauge Theory] Consider a pure $G$ gauge theory. The topological Wilson operators are precisely those corresponding to representations $\rho$ on which the connected identity component $G_0$ of the group acts trivially, or equivalently, representations of $\pi_0(G) = G/G_0$. 
\label{stat:topW}
\end{statement}

Now that we have described the topological Wilson lines, we should describe the endable Gukov-Witten operators in the pure gauge theory. Just as the non-topological Gukov-Witten operators are those that link nontrivially with endable Wilson lines, the non-topological Wilson lines are those that link nontrivially with endable Gukov-Witten operators. In order to see this, suppose we wish to define a Gukov-Witten operator on an open codimension-2 manifold $\mathcal{M}^{(d - 2)}$, with boundary $\partial \mathcal{M}$. To define such an operator, we delete a tubular neighborhood of $\mathcal{M}$, and define boundary conditions for our gauge field. In the bulk of $\mathcal{M}$, we place the Gukov-Witten boundary condition for some group element $g \in G$ on the transverse $S^1$. Along $\partial \mathcal{M}$, the transverse $S^1$ becomes the boundary of a hemisphere $\frac{1}{2} S^2$ that links with $\partial \mathcal{M}$, as depicted in Figure \ref{f3}. In order to define boundary conditions along $\partial \mathcal{M}$, we must pick a connection on $\frac{1}{2} S^2$ such that the holonomy around the boundary $S^1$ is our fixed group element $g$. By Lemma \ref{ambrose_singer}, this is only possible for $g \in G_0$, and in fact since our curvature may be arbitrary, this is possible for any $g \in G_0$. Thus, we learn:

\begin{statement}[Endable Gukov-Witten Operators in Pure Gauge Theory] Consider a pure $G$ gauge theory. The endable Gukov-Witten operators are precisely those corresponding to conjugacy classes contained in the identity component $G_0$ of $G$.
\label{stat:endableGW}
\end{statement}

What are the objects on which endable Gukov-Witten operators end, defined by gauge fields on a hemisphere? We may identify these codimension-3 operators as improperly quantized 't Hooft operators. These should not be confused with genuine 't Hooft operators, which are defined by boundary conditions on a closed $S^2$ rather than the hemisphere $\frac{1}{2}S^2$. Physically, in four dimensions, an 't Hooft operator may be thought of as the worldvolume of a probe monopole. An improperly quantized 't Hooft operator can be interpreted as the worldvolume of an ``illegal''  probe monopole, which does not respect Dirac quantization and, therefore, can only live at the boundary of a Gukov-Witten operator. 

\begin{figure}[!htb]
\begin{center}
\includegraphics[width=60mm]{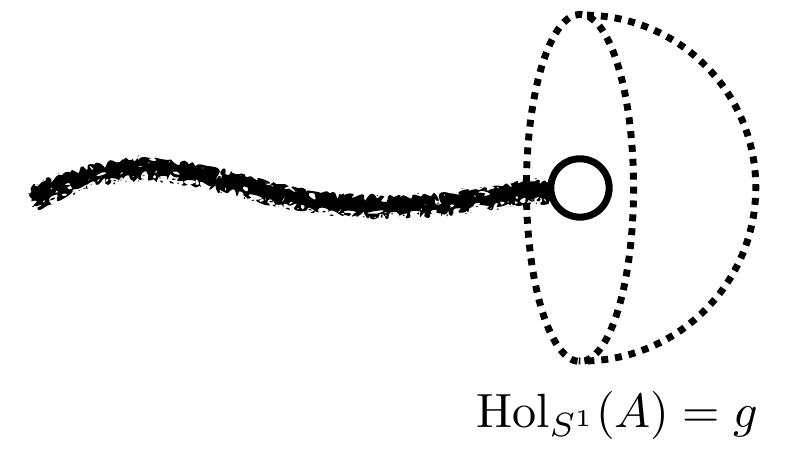}
\caption{A Gukov-Witten operator for a conjugacy class $[g]$ in the identity component of $G$ (represented by the fuzzy line in the picture) may end on an ``improperly quantized'' 't Hooft operator, depicted here by an open dot. To have a Gukov-Witten operator end, we specify a boundary condition for the gauge field on a hemisphere of the ending locus (depicted by dotted lines) by demanding that the holonomy becomes trivial in a smooth way. By Lemma \ref{ambrose_singer}, this is only possible for $g$ in the identity component. }
\label{f3}
\end{center}
\end{figure}

\subsubsection{'t Hooft Operators and Magnetic Global Symmetries}
\label{sec:tHO}

 For completeness (no pun intended), we can also construct the magnetic analogs of these objects in general dimension. The 't Hooft operator is constructed in a manner similar to Gukov-Witten operators, by excising a codimension-3 locus of spacetime and prescribing boundary conditions around it. 
Correlators of topological operators are sensitive only to the topological class of the boundary conditions, which is simply a $G$-bundle on the angular $S^2$. These are classified by the equatorial transition function, which is a map from $S^1$ to the gauge group, modulo conjugation. The topological class of a map  from $S^1$ to $G$ is the same as an element of $\pi_1(G)$. This element is invariant if we conjugate by elements in the connected component of the identity of $G$, which describe continuous deformations of the path. This is not the case when we conjugate by an element not in the identity component. The group of components $\pi_0(G)$ acts on the space of paths, so 't Hooft operators are naturally classified by classes in $\pi_1(G)$ modulo conjugation by $\pi_0(G)$.

In four dimensions, 't Hooft operators are line operators, and the definition of an 't Hooft line we have given above is essentially the same as the one originally given by 't Hooft \cite{tHooft:1977nqb}. It is worth noting, however, that the name ``'t Hooft line'' has been assigned to more than one concept in the literature \cite{Kapustin:2005py}. In \cite{Kapustin:2006pk}, for instance, ``'t Hooft line'' refers to ``Wilson line of the Langlands dual group,'' which derives from the GNO duality exchanging electric and magnetic charges \cite{Goddard:1976qe}. This notion leads to a finer classification of 't Hooft lines than the one we have given above. For instance, for $G=SU(2)$, the definition in the previous paragraph leads to the conclusion that pure $SU(2)$ theory has no 't Hooft lines, since $\pi_1(SU(2))=0$. However, the Langlands dual group, which is $SO(3)= SU(2)/\mathbb{Z}_2$, has nontrivial Wilson lines. 
This finer classification of 't Hooft lines has the drawback that it invokes the Langlands dual of $G$, which to our knowledge is not known for disconnected  $G$. We plan to explore this further in upcoming work.

\section{Compact, Connected Gauge Groups}\label{sec:Connected}

So far, we have primarily focused on topological and endable operators in pure gauge theories. We now want to see how the above story is modified in the presence of charged particles and dynamical, extended objects. We begin by studying compact, connected gauge groups, for which $G_0 = G$, and the centralizer $Z_G(G_0)$ is simply the center $Z(G)$. Hence, Statement~\ref{stat:topGW} implies that the topological Gukov-Witten operators in the pure gauge theory are labeled by elements of the center $Z(G)$. On the other hand, $G_0$ acts trivially only on the trivial representation of $G$, so Statement~\ref{stat:topW} implies that no Wilson operators are topological in the pure gauge theory. Relatedly, Statement~\ref{stat:endableGW} indicates that all Gukov-Witten operators are endable in the pure gauge theory, because $G_0 = G$.

\subsection{\texorpdfstring{$U(1)$}{U1} Gauge Theory}\label{sec:U1}
As a first example, we review $U(1)$ gauge theory in $d$ spacetime dimensions. Most of this discussion can be found in \cite{Gaiotto:2014kfa}.

Pure $U(1)$ gauge theory has action
\begin{equation}
S = -\frac{1}{2g^2} \int F \wedge {\star F}.
\end{equation}
This theory has Wilson line operators, each of which is supported on a closed 1-manifold and labeled by an electric charge $n \in \mathbb{Z}$, given by
\begin{equation}
W_{n}(\gamma) = \exp \left(i n \oint_\gamma A  \right).
\end{equation}
Such a Wilson line is charged under a 1-form $U(1)$ global symmetry, with conserved Noether current
\begin{equation}
\JE_2 = \frac{1}{g^2} {\star F}.
\end{equation}
The generators of this global symmetry are the Gukov-Witten operators. Since the gauge group is abelian and connected, its centralizer $Z_G(G_0)$ is the whole group, so by Statement~\ref{stat:topGW} every Gukov-Witten operator is topological. Each such operator is labeled by a phase $\alpha \in [0, 2 \pi)$ and supported on a codimension-2 manifold $\mathcal{M}^{(d-2)}$. Explicitly, such an operator is given by the integral of the Noether current over $\mathcal{M}^{(d-2)}$,
\begin{equation}
\label{UFe}
U_{g=\mathrm{e}^{i \alpha}} (\mathcal{M}^{(d-2)}) = \exp \left(  \frac{ i \alpha}{g^2} \int_{\mathcal{M}^{(d-2)}} {\star F} \right).
\end{equation}
These operators fuse according to the $U(1)$ group law,
\begin{equation}
U_{g}(\mathcal{M}^{(d-2)}) \times {U}_{g'}(\mathcal{M}^{(d-2)})  = {U}_{g''}(\mathcal{M}^{(d-2)}),
\end{equation}
with $g=\exp (i \alpha)$, $g' = \exp (i \beta)$, and $g'' = g g' = \exp (i (\alpha + \beta))$.
The Ward identity for this 1-form global symmetry implies that such a charge operator supported on an $S^{d-2}$ linking a Wilson line is equal to the Wilson line up to a phase,
\begin{equation}
U_{g=\mathrm{e}^{i \alpha}}(S^{d-2}) W_n(\gamma) = \mathrm{e}^{i \alpha n} W_n(\gamma).
\end{equation}
These Gukov-Witten operators can also be supported on codimension-2 manifolds with boundary. In this case, the boundary of such a manifold is an improperly quantized 't Hooft operator of magnetic charge $m = \alpha$. This means that all Gukov-Witten operators are endable, and relatedly there are no nontrivial topological Wilson lines.

Suppose now that instead of the pure gauge theory, we consider $U(1)$ gauge theory with a field $\psi$ of electric charge $N$. Now, the Wilson line of charge $N$ may be supported on a manifold with boundary $\gamma(x, y)$, with $\psi(x)$, $\psi^\dagger(y)$ supported on the endpoints.\footnote{More generally, the Wilson line of charge $k N$ may be supported on a manifold with boundary $\gamma(x, y)$, with $(\psi(x))^k$, $(\psi^\dagger(y))^k$ supported on the endpoints. If $\psi(x)$ is fermionic so that some power $(\psi(x))^k$ vanishes, we may seperate the operator insertions, and find operators of charge $k N$ in the OPE.} As a result, as illustrated in Figure \ref{fig:toplinking}, a surface operator that links nontrivially with such a Wilson line cannot be topological. In particular, the electric surface operators $U_{g=\mathrm{e}^{i \alpha}}$ that remain topological necessarily have $\alpha = 2 \pi k/N$ for some integer $k$: the other surface operators still exist, but they are no longer topological.

From this, we see that the $U(1)$ electric 1-form global symmetry in the pure gauge theory has been broken to a $\mathbb{Z}_N$ subgroup by the addition of the charge $N$ matter. For $N=1$, this 1-form electric symmetry is completely broken, and none of the electric surface operators $U_g$ are topological.

An analogous story applies to the 't Hooft operators in the theory. An 't Hooft operator $V_m(\cal{C})$ is labeled by an integer $m$ and supported on a closed codimension-3 manifold $\cal{C}$. Such operators are charged under a $(d-3)$-form $U(1)$ symmetry with conserved Noether current
\begin{equation}
\JM_2 = \frac{1}{2 \pi} F.
\end{equation}
The generators of this global symmetry are labeled by a phase $\eta \in [0, 2 \pi)$ and are supported on dimension-2 manifolds $\mathcal{M}^{(2)}$. Such an operator is given by the integral of the Noether current over $\mathcal{M}^{(2)}$,
\begin{equation}
\label{UFm}
\tilde{U}_{g=\mathrm{e}^{i \eta}}(\mathcal{M}^{(2)}) = \exp \left(  \frac{ i \eta}{2 \pi} \int_{\mathcal{M}^{(2)}} F \right).
\end{equation}

The Ward identity for the electric 1-form global symmetry implies that such a charge operator supported on an $S^2$ linking an 't Hooft operator is equal to the 't Hooft operator up to a phase,
\begin{equation}
\tilde{U}_{g=\mathrm{e}^{i \eta}}(S^2) V_m(\mathcal C) = \mathrm{e}^{i \eta m} V_m(\mathcal{C}).
\end{equation}
These charge operators can also be supported on 2-dimensional manifolds with boundary. In this case, the boundary of such a manifold is an improperly quantized Wilson line of electric charge $n = \eta$.

An 't Hooft operator of charge $M$ may end in the presence of a magnetically charged $(d-4)$-brane (i.e., a monopole), of magnetic charge $M$. As a result, a surface operator that links nontrivially with such an 't Hooft operator cannot be topological. In particular, the magnetic surface operators $\tilde{U}_{g=\mathrm{e}^{i \eta}}$ that remain topological necessarily have $\eta = 2 \pi k/M$ for some integer $k$: the other surface operators still exist, but they are no longer topological. The $U(1)$ magnetic $(d-3)$-form global symmetry in the pure gauge theory has been broken to a $\mathbb{Z}_M$ subgroup by the addition of the charge $M$ magnetic state. For $M=1$, this 1-form magnetic symmetry is completely broken, and none of the magnetic surface operators $\tilde U_g$ are topological.

In four dimensions, the fusion of a Wilson line of charge $n$ with an 't Hooft line of charge $m$ produces a dyonic line $L_{n,m}$ of charge $(n, m)$.
Similarly, electric and magnetic surface operators may be fused to give operators labeled by both an electric phase and a magnetic phase,
\begin{equation}
U_{(g = \mathrm{e}^{i \alpha}, g' = \mathrm{e}^{i \eta})}(\mathcal{M}^{(2)}) = \exp \left( \int_{\mathcal{M}^{(2)}} \frac{i \alpha}{g^2} (\star F)  + \frac{ i \eta}{2 \pi}  F \right)\,.
\end{equation}
Such surfaces link with dyonic Wilson/'t Hooft lines according to
\begin{equation}
U_{(g = \mathrm{e}^{i \alpha}, g' = \mathrm{e}^{i \eta})}(S^2) L_{n, m}(\gamma) = \mathrm{e}^{i \alpha n+ i \eta m}  L_{n, m}(\gamma) .
\end{equation}
Suppose we have a $U(1)$ gauge theory with dyons of charge $(N_1, M_1)$, $(N_2, M_2)$,...,$(N_l, M_l)$. Now, the only surfaces that remain topological will be the ones that link trivially with all of the dyons; namely, those $U_{g = \mathrm{e}^{i \alpha}, g' = \mathrm{e}^{i \eta}}$ satisfying
\begin{equation}
N_i \alpha + M_i \eta \equiv 0 \text{ (mod $2\pi$) for all $i= 1,...,l$}. 
\label{U1condition}
\end{equation}
Denoting the charge lattice of the $U(1)$ gauge theory by $\Gamma \simeq \mathbb{Z}^2$, the dyons in question will generate a sublattice $\Gamma_{\text{mat}}$. The 1-form global symmetry of the theory with the dyonic matter is then given by the Pontryagin dual of the quotient, $(\Gamma/\Gamma_{\text{mat}})^\vee$. (Note that $\mathbb{Z}^\vee = U(1), \mathbb{Z}_N^\vee = \mathbb{Z}_N$). Hence, the spectrum is complete if and only if there is no 1-form global symmetry.

Finally, let us note that all of the dimension-2 symmetry generators we have discussed are endable: $U_{(g=\mathrm{e}^{i \alpha}, g' = \mathrm{e}^{i \eta})}$ may end on an improperly quantized Wilson/'t Hooft line of charge $(\alpha, \eta)$. Relatedly, any line of charge $(n, m)$ will not be topological, as it links with at least one such surface. This is a general feature of 4d gauge theories with connected gauge group $G$: the Wilson/'t Hooft lines are not topological, whereas the surfaces that link with them are.

\subsection{\texorpdfstring{$SU(N)/\mathbb{Z}_K$}{SUNZK} Gauge Theory}\label{sec:SUNZK}

Let $K$ be a divisor of $N$. Pure $SU(N)/\mathbb{Z}_K$ gauge theory has Wilson lines labeled by representations of $SU(N)/\mathbb{Z}_K$. The adjoint Wilson line is endable, as it may end on a gluon. Likewise, any representation appearing in some tensor product of adjoint representations is endable. If $K = N$, tensor products of the adjoint generate the full set of irreducible representations, and the spectrum is complete. If $K \neq N$, the spectrum is incomplete, as some Wilson lines are not endable. Wilson lines fall into equivalence classes, each of which is specified by an integer modulo $N/K$, and a Wilson line labeled by an integer $n$ is endable if and only if $n=0$.

The Gukov-Witten operators are labeled by conjugacy classes of $SU(N)/\mathbb{Z}_K$. Since the adjoint Wilson line is endable, any Gukov-Witten operator that links nontrivially with the adjoint line is not topological. The topological operators in this case are all invertible, and they correspond precisely to 
elements of the center of $SU(N)/\mathbb{Z}_K$, which is given by $\mathbb{Z}_{N/K}$. If $N=K$, therefore, there are no such topological operators. If $K \neq N$, these topological operators furnish a $\mathbb{Z}_{N/K}$ 1-form symmetry, under which non-endable Wilson lines are charged.

If matter is added in a representation of the gauge group, additional Wilson lines may become endable, and the 1-form symmetry may be broken to a subgroup of $\mathbb{Z}_{N/K}$. In particular, if a Wilson line in the equivalence class specified by the integer $n$ mod $N/K$ is endable, then the 1-form symmetry is broken to $\mathbb{Z}_{\text{gcd}(n,N/K)}$. If all Wilson lines are endable, the 1-form symmetry will be broken completely. If some Wilson line remains non-endable after the addition of charged matter, it will be charged under some nontrivial remnant of the $\mathbb{Z}_{N/K}$ 1-form center symmetry.

All of the Gukov-Witten operators in $SU(N)/\mathbb{Z}_K$ gauge theory may end on improperly quantized 't Hooft operators. Relatedly, none of the Wilson lines are topological.

A similar story plays out on the magnetic side of things (at least with no discrete theta angles).
't Hooft operators of pure $SU(N)/\mathbb{Z}_K$ gauge theory are classified by elements of $\pi_1(SU(N)/\mathbb{Z}_K) = \mathbb{Z}_{N/K}$. Non-endable 't Hooft operators carry charge under a $\mathbb{Z}_{N/K}$ $(d-3)$-form symmetry.
In four dimensions, 't Hooft line operators admit a finer classification: they are in 1-1 correspondence with representations of the Langlands dual group ${}^L(SU(N)/\mathbb{Z}_K) = SU(N)/\mathbb{Z}_{N/K}$. The 't Hooft line labeled by the adjoint representation of the Langlands dual group is endable in the pure gauge theory. Any 't Hooft line which is not endable is charged under a magnetic $\mathbb{Z}_{N/K}$ 1-form symmetry. Note that this group is equal to the center of the Langlands dual group.

 't Hooft operators may end in the presence of magnetically charged $(d-4)$-branes (i.e., monopoles). Analogously to the electric case, endability of all 't Hooft operators is equivalent to the absence of a magnetic 1-form symmetry. Furthermore, the 't Hooft operators themselves are not topological.

As an illustrative example, let us specialize to the case of $N=2$. $K=1$. 
Representations of $SU(2)$ are labeled by non-negative integers $m = 2 j$, where $j$ is the spin of the representation, and the conjugacy classes are labeled by $\theta \in [0, \pi]$, where a representative of the conjugacy class $\theta$ is given by diag$(\exp (i \theta), \exp( - i \theta))$.
The center of $SU(2)$ is isomorphic to $\mathbb{Z}_2$ and consists of the elements $\theta = 0 , \pi$.
The characters are given by 
\begin{equation}
\chi_m(\theta)  =  \frac{\sin((m+1) \theta)}{\sin(\theta)}.
\end{equation}
Pure $SU(2)$ gauge theory has matter in the adjoint representation $m=2$. This means that the topological Gukov-Witten operators must have
\begin{equation}
\frac{\sin(3 \theta)}{\sin(\theta)} = \chi_{m=2}(\theta) = \chi_{m=2}(0) = 1,
\end{equation}
which is true for $\theta = 0, \pi$. We see that indeed, there is a $\mathbb{Z}_2$ 1-form symmetry, generated by the Gukov-Witten operator with $\theta = \pi$. This group elements with $\theta = 0, \pi$ are precisely those in the center of $SU(2)$, so as expected, the 1-form symmetry group is equivalent to the center. It is also clear that in the presence of matter in a representation with $m$ odd, none of the surfaces will remain topological. Simultaneously, the spectrum is complete: there is a (possibly multiparticle) state in every representation of the gauge group, since every representation of $SU(2)$ appears in tensor products of the $m$ representation and the adjoint representation, and there is no electric 1-form global symmetry remaining.

In pure $SO(3) = SU(2)/\mathbb{Z}_2$ gauge theory, on the other hand, the spectrum is already complete even without adding charged matter. At the same time, the center of $SO(3)$ is trivial, so there is no electric 1-form symmetry. There is, however, a magnetic $\mathbb{Z}_2$ $(d-3)$-form symmetry, under which non-endable 't Hooft operators are charged. If appropriate magnetically charged objects are added to the theory, this symmetry is broken, and relatedly the magnetic spectrum is complete.
In four dimensions, 't Hooft lines are labeled by representations of the Langlands dual group ${}^L SO(3) = SU(2)$, and the analysis of these lines, as well as the dimension-2 magnetic topological surfaces that link with them, follows immediately from our analysis of the $SU(2)$ gauge theory above. In particular, the $\mathbb{Z}_2$ magnetic 1-form symmetry is broken in the presence of monopoles in the fundamental representation of the Langlands dual group $SU(2)$.\footnote{In four dimensions, a similar story applies to $SO(3)$ gauge theory with a nontrivial discrete $\theta$ angle \cite{Aharony:2013hda}. Here, the operators charged under the $\mathbb{Z}_2$ magnetic 1-form symmetry are dyonic lines, and the 1-form symmetry is broken precisely when the dyonic spectrum is complete.}

\subsection{\texorpdfstring{$G_2$}{G2} Gauge Theory}

Before moving on to the general logic for arbitrary compact, connected groups, we illustrate how our above results cohere with some simple group theory facts for the particular case of the exceptional Lie group $G_2$. This group has trivial center, so by Statement~\ref{stat:topGW}, $G_2$ gauge theory has no topological Gukov-Witten operators, even in the absence of charged matter. In order for the general link between topological operators and completeness to hold, we then expect that pure $G_2$ gauge theory should already be electrically complete, i.e., every Wilson line must be endable. Translated into mathematics, we expect that every representation $\rho$ of $G_2$ should be contained in tensor products of the 14-dimensional adjoint representation with itself.

It is a well-known fact (see, e.g., Section 22.3 of~\cite{fulton2013representation}) that every irreducible representation of $G_2$ appears in some tensor power of the ``standard'' 7-dimensional representation of $G_2$ with itself. From, e.g., \cite{Yamatsu:2015npn}, we further have
\begin{align}
\mathbf{14} \otimes \mathbf{14} =  \mathbf{1}  \oplus \mathbf{14} \oplus \mathbf{27} \oplus  \mathbf{77} \oplus \mathbf{77}' \,,~~~\mathbf{14} \otimes \mathbf{27}  = \mathbf{7} \oplus \mathbf{14} \oplus \mathbf{27}  \oplus  \mathbf{64} \oplus   \mathbf{77} \oplus \mathbf{189} 
\end{align}
Therefore, the 7-dimensional irrep appears in the tensor product $\mathbf{14}^{\otimes 3}$ of three copies of the adjoint representation. This establishes that, indeed, all Wilson lines are endable in pure $G_2$ gauge theory.

\subsection{General Story for Compact, Connected Groups}\label{ssec:GeneralConnected}

\subsubsection{Electric Completeness}

As discussed in Section~\ref{sec:topGW}, in a pure $G$ gauge theory, the topological Gukov-Witten operators are labeled by elements of the center $Z(G)$. In the presence of charged matter in some representation $\rho$, the center symmetry is generically broken to a subgroup. In particular, an open Wilson line in the representation $\rho$ may be unlinked from the Gukov-Witten operator $U_z$, $z \in Z(G)$, as in Figure \ref{fig:toplinking}, which implies that $U_z$ remains a topological symmetry generator only if the representation $\rho$ transforms trivially under $z$. In other words, the 1-form symmetry remaining in the presence of the charged matter is given by the kernel of $\rho$.

If the spectrum of the gauge theory is incomplete, then $\rho$ cannot be a faithful representation of $G$: there necessarily exists some $z \neq 1 \in Z(G)$ such that $z$ is in the kernel of $\rho$ for all endable Wilson lines $\rho$.\footnote{We will prove this statement when we consider the case of a general compact Lie group in Section \ref{sec:General}. } As a result, the remaining 1-form symmetry is a nontrivial subgroup of the center. If the spectrum is complete, on the other hand, then there exists matter in some faithful representation $\rho$, whose kernel is trivial. As a result, the 1-form center symmetry is completely broken. There is thus a 1-1 correspondence between completeness of the gauge theory spectrum and the absence of an electric 1-form global symmetry.

\subsubsection{Magnetic Completeness}
\label{sec:connectedmagnetic}

In addition to the electric 1-form global symmetry discussed above, there is also a magnetic $(d-3)$-form global symmetry. In pure gauge theory, the 't Hooft operators charged under this global symmetry are labeled by elements of $\pi_1(G)$, and the associated symmetry group is given by $\pi_1(G)^\vee$, the Pontryagin dual of $\pi_1(G)$. In the presence of magnetically-charged $(d-4)$-branes (i.e., monopoles), some of these 't Hooft operators can end. When they do, any symmetry generators that link nontrivially with them will no longer be topological, and the $(d-3)$-form $\pi_1(G)^\vee$ global symmetry will be explicitly broken to a subgroup. The set of endable 't Hooft operators forms a subgroup $H$ of $\pi_1(G)$, which is normal since $\pi_1(G)$ is abelian. The remaining magnetic global symmetry is given by $(\pi_1(G)/H)^\vee$, the Pontryagin dual of ${\pi}_1(G)/H$. We thus see that if every 't Hooft operator is endable, so the magnetic spectrum is complete, then the magnetic 1-form symmetry will be completely broken. If there is at least one non-endable 't Hooft operator, so the monopole spectrum is incomplete, then some nontrivial subgroup of the magnetic 1-form symmetry will remain. This statement also follows from the relationship between electric completeness and the absence of an electric 1-form global symmetry applied to the Langlands dual group ${}^L G$.

\subsubsection{Twist Vortex Completeness}

Recall from Statement~\ref{stat:topW} in Section \ref{sec:topW} that the topological Wilson operators in pure $G$ gauge theory are precisely those corresponding to representations $\rho$ on which the identity component of the group acts trivially. For $G$ connected, the identity component is the entirety of $G$, and the only representation which acts trivially on $G$ is the trivial representation. Thus, there are no topological Wilson lines, even in the pure gauge theory.

Relatedly, by Statement~\ref{stat:endableGW}, every Gukov-Witten operator can end on an improperly quantized 't Hooft operator when $G$ is connected. This means that this theory cannot admit twist vortices (which correspond to endable Gukov-Witten operators {\em not} arising from elements of $G_0$), so twist vortex completeness is trivially satisfied. We will elaborate on this point when we consider more general gauge groups in Section \ref{sssec:twistvortexcomp} below.

\section{Finite Gauge Groups}\label{sec:Discrete}

The next case we consider is that of finite gauge groups, as previously discussed in detail in \cite{Rudelius:2020orz}. For finite gauge groups, $G_0$ is simply the identity element of the group. As a result, $Z_G(G_0) = G$, so Statement~\ref{stat:topGW} implies that in the pure gauge theory, all Gukov-Witten operators are topological. Similarly, because $G_0$ acts trivially on all representations, Statement~\ref{stat:topW} implies that all Wilson operators are topological in the pure gauge theory.

\subsection{\texorpdfstring{$\mathbb{Z}_N$}{ZN} Gauge Theory}\label{sec:ZN}

$\mathbb{Z}_N$ gauge theory admits a Lagrangian description as a BF-theory:
\begin{equation}
\mathcal{L} = \frac{N}{2\pi} B_{d-2} \wedge \rmd A_1\,.
\end{equation}
This theory has Wilson line operators of the form
\begin{equation}
W_n(\gamma) = \exp \left( i n \oint_\gamma A_1  \right)\,,~~~n = 0, 1, ..., N-1\,,
\end{equation}
with $V_n$ the operator of charge $n$. There are also codimension-2 Wilson surfaces for $B_{d-2}$:
\begin{equation}
\label{UB}
U_m(\mathcal{M}^{(d-2)}) = \exp \left( i m \oint_{\mathcal{M}^{(d-2)}} B_{d-2}  \right)\,,~~~m = 0, 1, ..., N-1\,,
\end{equation}
These Wilson surfaces also serve as the Gukov-Witten operators for the $\mathbb{Z}_N$ gauge theory.

All of these Wilson operators are topological: the Wilson lines generate a $\mathbb{Z}_N$ $(d-2)$-form global symmetry under which the Wilson surfaces are charged. Conversely, the Wilson surfaces generate a $\mathbb{Z}_N$ 1-form global symmetry under which the Wilson lines are charged. The Ward identity for this 1-form symmetry leads to a nontrivial linking of the Wilson lines and the Wilson surfaces: 
\begin{equation}
U_m(S^{d-2}) \cdot W_n(\gamma) = \mathrm{e}^{2 \pi i m n / N} W_n(\gamma)\,.
\end{equation}

Suppose we now add a particle of charge $K$ to the theory. The Wilson lines of charge $nK$, $n \in \mathbb{Z}$ are now endable. As a result, any surface $U_ m$ that links nontrivially with these lines is no longer topological. In particular, the surfaces $U_m$ that remain topological satisfy $m K = 0$ (mod $N$), which means that the 1-form global symmetry is reduced from $\mathbb{Z}_N$ to $\mathbb{Z}_{\text{gcd}(N,K)}$. In particular, the spectrum is complete if and only if there do not exist any topological Gukov-Witten operators.

$\mathbb{Z}_N$ gauge theory can also admit twist vortices, dynamical codimension-2 objects around which $\mathbb{Z}_N$-charged particles acquire Aharonov-Bohm phases. In the 4d context, these are dynamical strings \cite{Krauss:1988zc}. They carry magnetic flux under $A_1$. As discussed in \S\ref{sec:Top}, Gukov-Witten operators may be thought of as inserting a probe vortex. Hence, Gukov-Witten operators can end on codimension-3 operators that create a dynamical twist vortex, just as Wilson line operators can end on local operators that create a dynamical charged particle. In the $\mathbb{Z}_N$ gauge theory, Wilson lines are topological if Gukov-Witten operators cannot end, and generate a $\mathbb{Z}_N$ $(d-2)$-form global symmetry. The existence of twist vortices explicitly breaks this symmetry; a complete spectrum of twist vortices fully breaks it, rendering all of the Wilson lines non-topological.

\subsection{\texorpdfstring{$S_3$}{S3} Gauge Theory}\label{ssec:S3}

The symmetric group $S_3$ has three irreducible representations: the trivial representation $\mathbf{1}$, the sign representation $\mathbf{1}_-$ of dimension 1, and the standard representation $\mathbf{2}$ of dimension 2. It has three conjugacy classes: the trivial conjugacy class $[1]$, the conjugacy class $[\theta]$ of size 2, and the conjugacy class $[\tau]$ of size 3.

Correspondingly, the theory has three irreducible Wilson lines, $W_0$, $W_-$, and $W_2$, associated respectively with the irreducible representations of $S_3$. It has three codimension-2 Gukov-Witten operators, $T_{[1]}$, $T_{[\theta]}$, and $T_{[\tau]}$. All of these operators are topological. Unlike the previous examples we have seen, however, most of these topological operators do not generate global symmetries. The lone exception is the Wilson line $W_-$, which squares to the trivial operator and thus generates a $\mathbb{Z}_2$ $(d-2)$-form global symmetry. The nontrivial Wilson lines $W_-$, $W_2$ obey the fusion algebra
\begin{equation}
W_2 \times W_2 =  W_0 + W_- + W_2 \,,~~~~~W_2 \times W_- =  W_2\,.
\label{S3fusion}
\end{equation}
Note that $W_2$ does not have an inverse operator---it is a \emph{non-invertible} topological operator.
The nontrivial Gukov-Witten operators $T_{[\theta]}$, $T_{[\tau]}$ are similarly non-invertible.

The Wilson lines and the Gukov-Witten operators link nontrivially. In particular, we have
\begin{equation}
T_a(S^{d-2}) \cdot W_b(\gamma) = B_b(a) W_b(\gamma)\,,
\end{equation}
with $B_b(a)$ given by:
\begin{align}
B_b(a) = \begin{array}{c|ccc}
\text{b \textbackslash a} & [1] & [\tau] & [\theta] \\\hline
0 &1 & 3& 2 \\
- & 1& -3& 2 \\
$2$ & 1& 0 & -1  
\end{array}\,.
\end{align}

In the presence of matter in the sign representation of $S_3$, the Wilson line $W_-$ is endable. Since $B_-({[\tau]}) \neq B_0({[\tau]})$, the Gukov-Witten operator $T_{[\tau]}$ links nontrivially with an endable Wilson line, and therefore it is not topological in the presence of this charged matter.

In the presence of matter in the standard representation of $S_3$, all of the Wilson lines are endable by \eqref{S3fusion}, and none of the Gukov-Witten operators are topological. We see that completeness of the gauge theory spectrum is equivalent to the absence of topological codimension-2 Gukov-Witten operators.

Nonabelian finite-group gauge theories can also admit twist vortices, which have nonabelian Aharonov-Bohm interactions with charged matter \cite{Alford:1989ch,Preskill:1990bm}. In pure $S_3$ gauge theory, the Wilson lines are topological. However, if the Gukov-Witten operators are endable due to the existence of dynamical twist vortices, then Wilson lines with which they link nontrivially will no longer be topological. From the table of $B_b(a)$, we see that if the Gukov-Witten operator $T_{[\tau]}$ is endable, then both the $W_-$ and $W_2$ Wilson lines are not topological. On the other hand, if only the Gukov-Witten operator $T_{[\theta]}$ is endable, then $W_-$ will remain topological, because it has the same linking with $T_{[\theta]}$ as the trivial Wilson loop $W_0$. Because the equivalence class $[\tau]$ includes two-element transpositions, it generates the entire group $S_3$. Hence, endability of $T_{[\tau]}$ implies endability of $T_{[\theta]}$, because $T_{[\theta]}$ can arise in the fusion of $T_{[\tau]}$ operators. The converse is not true: representatives of the class $[\theta]$ do not generate all of $S_3$. Thus, we see that in this example, the absence of topological Wilson lines is equivalent to the endability of all nontrivial Gukov-Witten operators, which is in turn equivalent to the existence of a complete spectrum of twist vortices. We will argue below that this generalizes to all finite groups.

\subsection{General Story for Finite Groups}
\label{sec:generalfinite}

\subsubsection{Electric Completeness}

The following is a quick review, without proof, of completeness and topological operators in discrete gauge theories; for a general proof, see Section~\ref{general_electric_completeness_section}. More details can be found in \cite{Rudelius:2020orz}.

Consider the $d$-dimensional gauge theory of some (abelian or nonabelian) discrete gauge group $G$. Such a theory will have Wilson lines, each labeled by a representation of $G$, and codimension-2 Gukov-Witten operator, each labeled by a conjugacy class $[g]$ of $G$. For an abelian group, each conjugacy class consists of a single element, $[g] = g$, and every Gukov-Witten operator is invertible. For a nonabelian group, a conjugacy class will generically have more than one element.

A Gukov-Witten operator labeled by the conjugacy class $[g]$ can link with Wilson line labeled by a representation $\nu$. Namely, surrounding the Wilson line with an $S^{d-2}$ supporting a topological Gukov-Witten operator and shrinking the $S^{d-2}$ to a point yields the Wilson line times a linking coefficient $B_\nu([g])$:
\begin{equation}
B_\nu([g]) = \frac{\chi_\nu([g])}{ \chi_\nu(1) } \text{size}([g]),
\end{equation}
with $\chi_\nu$ the character of the representation $\nu$, $\text{size}([g])$ the number of elements in the conjugacy class $[g]$, and $[1]$ the trivial conjugacy class.

If a given Wilson line $\nu$ an end on some charged particle, some surfaces will no longer be topological. In particular, any surface which links nontrivially with $\nu$ will be rendered non-topological. To be more precise, if
\begin{equation}
\chi_{\nu}([g]) \neq  \chi_\nu(1) := \text{dim}(\nu),
\end{equation}
then the surface labeled by the conjugacy class $[g]$ will cease to be topological in the presence of matter in the $\nu$ representation of $G$. If the spectrum is complete, so that all Wilson lines can end, then every Gukov-Witten operator is rendered non-topological \cite{Rudelius:2020orz}. Conversely, if the spectrum is incomplete, there will exist at least one topological Gukov-Witten operator.

\subsubsection{Twist Vortex Completeness}

Just as the endability of all Wilson lines is equivalent to the non-existence of topological Gukov-Witten operators, the endability of all Gukov-Witten operators is equivalent to the non-existence of topological Wilson lines. In the special case of three dimensions, where Gukov-Witten operators are line operators, this follows from general properties of the modular tensor category of line operators \cite{Rudelius:2020orz, muger2003structure}. Here, we will provide an argument in higher dimensions.

The Wilson lines of a finite-group gauge theory are all topological in pure gauge theory. If twist vortices are coupled to the theory, however, their Gukov-Witten operators will become endable. If a Wilson line links nontrivially with an endable Gukov-Witten operator, it will not be topological, as illustrated in Figure~\ref{fig:toplinking} but with the role of Wilson lines and Gukov-Witten operators interchanged. A Wilson line in the representation $\rho$ links nontrivially with a Gukov-Witten operator for the conjugacy class $[g]$ when $\chi_\rho([g]) \neq \chi_\rho(1)$.  If all Gukov-Witten operators are endable, then there are no nontrivial topological Wilson lines, because the only representation $\rho$ for which $\chi_\rho([g]) = \chi_\rho(1)$ for all conjugacy classes $[g]$ is the trivial representation. 

The converse statement is less obvious: if there are no nontrivial topological Wilson lines, then all Gukov-Witten operators must be endable. The statement that there are no nontrivial topological Wilson lines means that for every nontrivial representation $\rho$, there exists at least one conjugacy class $[g]$ corresponding to an {\em endable} Gukov-Witten operator for which $\chi_\rho([g]) \neq \chi_\rho(1)$. Equivalently, given a representation $\rho$ for which $\chi_\rho([g]) = \chi_\rho(1)$ for all $[g]$ for which the Gukov-Witten operator is endable, it must be the trivial representation. To show that this implies that {\em all} Gukov-Witten operators are endable, we use the following fact about group theory:

\begin{lemma}
Consider a group $G$ and a set $S \subseteq \mathrm{Conj}(G)$ of conjugacy classes of $G$, with the property that $\chi_\rho([g]) = \chi_\rho(1)$ for all $[g] \in S$ only if $\rho$ is the trivial representation of $G$. Then the representatives of the conjugacy classes in $S$ generate the entire group $G$.
\label{lemmafinite}
\end{lemma}

\begin{proof}
The representatives of the conjugacy classes in $S$ generate a subgroup $H$ of $G$. Because it is constructed from conjugacy classes, this subgroup is normal: $H \trianglelefteq G$. To show that $H = G$, we will argue that the quotient $G/H$ is trivial. Consider an irreducible representation $\nu$ of $G/H$. Composition of $\nu$ with the projection $\pi: G \to G/H$ determines a pullback representation $\rho$ of $G$, i.e., $\rho = \pi^* \nu := \nu \circ \pi$. For any $h \in H$, we have $\rho(h) = \nu(\pi(h)) = \nu(1) = 1$. In particular, this means that $\chi_\rho([g]) = \chi_\rho(1)$ for all $[g] \in S$. But then, by assumption, $\rho$ is the trivial representation of $G$. Because $\pi$ is surjective, we conclude that $\nu$ is the trivial representation of $G/H$. Hence $G/H$ admits only trivial representations, so it is the trivial group, and $H = G$.
\end{proof}

If there are no topological Wilson lines, there must be a sufficient set $S$ of endable Gukov-Witten operators to link nontrivially with all possible Wilson lines. The theorem tells us that this set of endable Gukov-Witten operators can generate {\em all} Gukov-Witten operators through their fusion. But any operator appearing in the fusion of endable operators must also be endable, and so we conclude that all Gukov-Witten operators are endable. Gukov-Witten operators for finite groups end on operators that create dynamical twist vortices (in 4d, these are twist strings). We conclude that the absence of topological Wilson line operators is equivalent to completeness of the spectrum of twist vortices.

In fact, nothing in this argument relied on $G$ being a finite group. Thus, we have established one of the main results of our paper, Statement~\ref{stat:TwistVortexCompleteness}, as promised in Section~\ref{sec:intro}. When $G$ is a general compact group, endability of some Gukov-Witten operators will correspond to the existence of twist vortices, while others will already be endable by Statement~\ref{stat:endableGW}. We will elaborate on this point in more detail below, in Section~\ref{sssec:twistvortexcomp}.

\subsubsection{Magnetic Completeness}

The magnetic spectrum of a discrete gauge theory is trivially complete: $\pi_1(G)$ is trivial, since the identity component of $G$ is simply a point. Thus, there are no 't Hooft operators, and there is no magnetic $(d-3)$-form symmetry.

\subsubsection{A Surface Operator Subtlety}\label{ssec:Z2Z2}

Before we move on from finite groups, let us mention a subtlety in our arguments, previously discussed in Section~5.1 of Ref.~\cite{Rudelius:2020orz}.\footnote{We thank Shu-Heng Shao, Po-Shen Hsin, Meng Cheng, and Qing-Rui Wang for bringing this subtlety to our attention, and for discussions on this topic.} In $\mathbb{Z}_2 \times \mathbb{Z}_2$ gauge theory in any number of dimensions, there is a surface operator $v(\Sigma)$ associated with the nontrivial element of $H^2(\mathbb{Z}_2 \times \mathbb{Z}_2, U(1)) \cong \mathbb{Z}_2$, which has a nontrivial {\em triple} linking with two Gukov-Witten operators. This suggests a different way that a Gukov-Witten operator can cease to be topological: if we render the operator $v(\Sigma)$ endable by adding a new dynamical string to the theory, then the Gukov-Witten operators that link with $v$ would no longer be topological. This poses a loophole, through which the absence of topological Gukov-Witten operators might not necessarily imply the endability of all Wilson line operators. Statement~\ref{stat:ElectricCompleteness} was artfully phrased to avoid this loophole, by referring only to $G$ gauge theories coupled to matter fields in representations of $G$ (and not to the stringlike objects that could render $v(\Sigma)$ endable). 

Of course, in this setting one might also inquire about the necessary ingredients to eliminate the topological operator $v(\Sigma)$ from the theory. We believe that a sufficiently careful analysis will show that, in general, the endability of {\em all} extended operators is equivalent to the absence of {\em any} topological operators. However, the precise proofs that we present in this paper address only limited subsets of operators and dynamical objects. We will discuss similar subtleties in the context of Chern-Simons terms in more detail in Section~\ref{sec:CS}. We leave a complete exploration of the details of $H^2(G,U(1))$ surfaces, and the search for a universal proof of the equivalence of completeness and the absence of topological operators, for future work.

\section{Compact, Disconnected Gauge Groups}\label{sec:Disconnected}

We will now move on to study topological operators for gauge theories for the case of a generic, disconnected Lie group.\footnote{We do not require our Lie groups to have dimension greater than zero, i.e., we include finite groups as examples of compact Lie groups.} This case has features in common with both those of compact, connected Lie groups, and of finite groups. Like the latter, it admits topological Wilson lines as well as topological Gukov-Witten operators. In fact, the presence of non-invertible symmetries for such a gauge theory was recently studied \cite{Nguyen:2021yld}, for the gauge group $U(1)^{N - 1} \rtimes S_N$ (our main example of $O(2)$ gauge theory is the case $N = 2$).

\subsection{\texorpdfstring{$O(2)$}{O2} Gauge Theory}\label{sec:O2}

$O(2) = U(1) \rtimes \mathbb{Z}_2$ gauge theory can be constructed from $U(1)$ gauge theory by gauging charge conjugation \cite{Kiskis:1978ed, Schwarz:1982ec}. 

$O(2)$ has a trivial irrep $\mathbf{1}$, a ``det'' irrep $\mathbf{1}_\mathrm{det}$ of dimension 1 (which gives the determinant of an element of $O(2)$), and other irreps of dimension 2 labeled by positive integers $q \geq 1$, which we will denote  $\mathbf{2}_q$. The det rep is the adjoint representation, reflecting the fact that $F_{\mu \nu}$ is not gauge invariant but transforms to $-F_{\mu \nu}$ under charge conjugation. Under the branching $O(2) \rightarrow U(1)$, $\mathbf{2}_q \rightarrow \mathbf{1}_q \oplus \mathbf{1}_{-q}$, whereas $\mathbf{1}_\mathrm{det} \rightarrow \mathbf{1}$, the trivial rep. Thus, a charge $q$ Wilson line of $O(2)$ gauge theory can be constructed as a gauge-invariant sum of Wilson lines of $U(1)$ gauge theory:
\begin{equation}
W_{q}^{O(2)}(\gamma) = W_{q}^{U(1)}(\gamma) + W_{-q}^{U(1)}(\gamma) =  \exp\left( iq \oint_\gamma A\right) + \exp\left( -iq \oint_\gamma A\right).
\end{equation}

The fusion between two $O(2)$ Wilson lines is given by the tensor product of their representations. In particular, we have (for $q \neq q'$)
\begin{align}
\begin{split}
\mathbf{2}_q \otimes \mathbf{2}_{q'} = \mathbf{2}_{q+q'} \oplus  \mathbf{2}_{|q-q'|}, &\quad
\mathbf{2}_q \otimes \mathbf{2}_q = \mathbf{2}_{2q} \oplus \mathbf{1}_\mathrm{det} \oplus \mathbf{1}, \\
\mathbf{2}_q \otimes \mathbf{1}_\mathrm{det} = \mathbf{2}_q, &\quad
\mathbf{1}_\mathrm{det} \otimes \mathbf{1}_\mathrm{det} = \mathbf{1}.
\end{split}
\end{align}
We can understand the first of these fusion laws from the perspective of the $U(1)$ Wilson lines:
\begin{align}
\begin{split}
W_{q}^{O(2)} \times W_{q'}^{O(2)} & =  \exp\left( i (q+q') \oint A\right) + \exp\left( i (q-q') \oint A\right)  \\
&+ \exp\left( i (-q+q') \oint A\right)+  \exp\left(- i (q+q') \oint A\right) \\
&=   W_{q+ q'}^{O(2)} + W_{|q-q'|}^{O(2)}.
\end{split}
\end{align}

As usual, codimension-2 Gukov-Witten operators are labeled by conjugacy classes of $O(2)$. All elements of $O(2)$ with determinant $-1$, (i.e., those elements disconnected from the identity) lie in a single conjugacy class. The associated Gukov-Witten operator $T^{O(2)}_\text{disc}$ links nontrivially with the adjoint (det) line, which is endable, so this surface operator is not topological.

The remaining conjugacy classes have determinant $+1$ and are labeled by an angle $\theta \in [0, \pi]$, where a given class $\theta$ may be represented by the matrix
\begin{equation}
\left(
\begin{array}{cc}
 \cos \theta & - \sin \theta \\
\sin \theta & \cos \theta 
\end{array}
\right),
\end{equation} 
and $\theta \sim - \theta \sim \theta + 2 \pi$. The center of $O(2)$ is isomorphic to $\mathbb{Z}_2$ and is generated by the element with $\theta = \pi$.

The Gukov-Witten operators labeled by conjugacy classes $\theta \in (0, \pi)$ may be written as gauge-invariant sums of Gukov-Witten operators of $U(1)$ gauge theory:
\begin{align}
\begin{split}
T_{\mathrm{e}^{i \theta}}^{O(2)}(\mathcal{M}^{(d-2)}) &=U_{\mathrm{e}^{i \theta}}^{U(1)}(\mathcal{M}^{(d-2)}) + U_{\mathrm{e}^{-i \theta}}^{U(1)}(\mathcal{M}^{(d-2)}) \\ &= \exp\left( \frac{i \theta}{g^2} \oint_{\mathcal{M}} \star F\right) + \exp\left( - \frac{i \theta}{g^2} \oint_{\mathcal{M}} \star F \right).
\label{O2GWsplitting}
\end{split}
\end{align}
All of these operators are topological in pure $O(2)$ gauge theory, and they have quantum dimension 2.  Meanwhile, the surface with $\theta = \pi$ is given simply by the $\theta = \pi$ surface of the $U(1)$ gauge theory:
\begin{equation} 
T_{\mathrm{e}^{i \pi}}^{O(2)}(\mathcal{M}^{(d-2)}) =U_{\mathrm{e}^{i \pi}}^{U(1)}(\mathcal{M}^{(d-2)}) = \exp\left(  \frac{i \pi}{g^2}\oint_{\mathcal{M}} \star F\right)  .
\end{equation}
This surface has quantum dimension 1 and generates the $\mathbb{Z}_2$ ``center'' 1-form global symmetry.

The fusion of two topological surfaces labeled by $\theta$, $\theta'$ can be understood in terms of the $U(1)$ electric symmetry generators. In particular, for $\theta \neq \theta'  \neq \pi - \theta' \in (0 , \pi)$, we have
\begin{align}
T^{O(2)}(\theta) \times  T^{O(2)}(\theta') &= \exp\left( \frac{i (\theta+ \theta')}{g^2} \oint \star F\right) +  \exp\left( \frac{i (\theta- \theta')}{g^2} \oint \star F\right) \nonumber \\ & + \exp\left(  \frac{i (-\theta+ \theta')}{g^2} \oint \star F\right)+\exp\left( - \frac{i (\theta+ \theta')}{g^2} \oint \star F\right) \\
&=  T^{O(2)}(\theta + \theta') \times  T^{O(2)}(\theta-\theta')  \nonumber.
\end{align}
Similarly, we have for $\theta \in (0, \pi)$:
\begin{align}
T^{O(2)}(\theta) \times  T^{O(2)}(\pi) &= \exp\left( \frac{i (\theta+ \pi)}{g^2} \oint \star F\right) +  \exp\left( \frac{i (-\theta+ \pi)}{g^2} \oint \star F\right)=  T^{O(2)}(\theta + \pi) ,
\end{align}
and $ T^{O(2)}(\pi) \times  T^{O(2)}(\pi) = 1$.

The fusion of two surfaces with $\theta = \theta'$ or $\theta = \pi - \theta'$ are slightly more complicated. The former is 
\begin{align}
T^{O(2)}(\theta) \times  T^{O(2)}(\theta) = 1 + W_\text{det} + T^{O(2)}(2 \theta),
\end{align}
while the latter is
\begin{align}
T^{O(2)}(\theta) \times  T^{O(2)}( \pi - \theta) =  T^{O(2)}(\pi) + T^{O(2)}(\pi, \text{det})+ T^{O(2)}(2 \theta - \pi),
\end{align}
Here, 1 is the trivial surface, $W_\text{det}$ is a trivial surface with a Wilson line in the det representation, and $T^{O(2)}(\pi, \text{det})$ is a $\pi$ surface with a det Wilson line. The appearance of the det Wilson line---an operator of lower dimension---at the junction of higher-dimensional operators is characteristic of a higher-group global symmetry \cite{Benini:2018reh}. In the case at hand, two of the higher-dimensional operators are not invertible, but rather have quantum dimension 2.\footnote{ The appearance of electrically charged particles in the fusion of strings has been discussed previously in the context of discrete gauge theories \cite{Moradi_2015}. It would be worthwhile to explore further the analogous fusion of surface operators we have seen here.}

The fact that $W_\mathrm{det}$ appears in the fusion of two topological surfaces indicates that det is a topological line. This can also be seen as a consequence of Statement~\ref{stat:topW}, because the identity component of $O(2)$ acts trivially on the determinant representation.

If we surround a Wilson line in a representation $\nu$ of $O(2)$ with a topological surface $T_{\mathrm{e}^{i \theta}}$ and shrink the surface to a point, we find the Wilson line times a linking coefficient,
\begin{equation}
B_\nu(\theta) = \frac{\chi_\nu(\theta) }{\chi_\nu(0) } \times \text{size}(\theta),
 \label{O2linking}
\end{equation}
with size$(\theta) =1$ for $\theta =0, \pi$, and size$(\theta) =2$ otherwise, and
\begin{align}
\chi_q(\theta) = \mathrm{e}^{i q \theta} + \mathrm{e}^{-i q \theta} = 2 \cos (q \theta) ,~~~\chi_0(\theta) = \chi_{\text{det}}(\theta) = 1.
\label{O2characters}
\end{align}

In the presence of matter in a representation $q$, Wilson lines in the $q$ representation can end. The surfaces that remain topological are those that link trivially with such a Wilson line, i.e., those satisfying $\chi_q(\theta) = \chi_q(0)$. From \eqref{O2characters}, this is equivalent to $\cos(q \theta) = 1$. We see that for $q=1$, the only such surface that survives is the trivial surface: as expected, no nontrivial surfaces remain topological when the spectrum is complete.

On the other hand, for $q=2$, the $\mathbb{Z}_2$ surface $T(\theta  = \pi)$ remains topological. The ``center symmetry'' survives even in the presence of matter in a representation $q$ provided that $q$ is even.

For more general $q$, however, there will be additional surfaces that remain topological, indexed by an integer $k$:
\begin{equation}
\theta = 2 \pi k/q\,,0<\theta \leq \pi.
\end{equation}
These surfaces are not invertible for $\theta \in (0, \pi)$. This is evidenced by the fact that the conjugacy class $\theta$ has two elements for $\theta \in (0,\pi)$, so the quantum dimension of $T(\theta)$ is equal to 2. Note that for $q=3$, the invertible operator $(\pi, 0)$ will cease to be topological. Nonetheless, there are still topological operators with $\theta = 2\pi/3$. We see that topological Gukov-Witten operators exist if and only if the spectrum is incomplete.

Similarly, topological Wilson lines exist if and only if the spectrum of twist vortices is incomplete. The one Wilson line which is topological in pure $O(2)$ gauge theory is the det line. It links nontrivially with the Gukov-Witten operator associated with the conjugacy class of elements with determinant $-1$. This Gukov-Witten operator becomes endable, and the det line non-topological, when the theory contains a dynamical twist vortex which induces charge conjugation on a charged particle that circles the vortex. In the 4d case, this twist vortex is the familiar ``Alice string'' of $O(2)$ gauge theory \cite{Schwarz:1982ec,Alford:1990mk,Preskill:1990bm}. All other Gukov-Witten operators can already end on improperly quantized 't Hooft operators in the pure gauge theory, according to Statement~\ref{stat:endableGW}. This property is inherited from the endability of the Gukov-Witten operators of $U(1)$ gauge theory.

Let us now specialize to 4 dimensions and study the 't Hooft operators (now lines) of the theory. Like the Wilson lines, these are labeled by positive integers and can be constructed from gauge-invariant sums of the 't Hooft lines of $U(1)$ gauge theory prior to gauging charge conjugation:
\begin{equation}
V_{q}^{O(2)}(\gamma) = V_{q}^{U(1)}(\gamma) + V_{-q}^{U(1)}(\gamma).
\label{O2Wilsonsplitting}
\end{equation}
More mathematically, the 't Hooft lines are in 1-1 correspondence with elements of the quotient $\pi_1(O(2))/\pi_0(O(2))$, where the basepoint of loops in $\pi_1(O(2))$ is the identity, and elements of $\pi_0(O(2))$ act by conjugation on these loops. $\pi_1(O(2))$ is then equal to $\pi_1(U(1)) \cong \mathbb{Z}$ (since $U(1)$ is the connected component of $O(2)$) and the quotient by $\pi_0(O(2))$ identifies loops of opposite orientation.

The fusion algebra of these loops is identical to the fusion of $O(2)$ Wilson lines: lines of magnetic charge $q$, $q'$ fuse to lines of charge $q + q'$, $|q-q'|$. Two 't Hooft lines of the same charge $q$ fuse according to
\begin{equation}
V_{q}^{O(2)} \times V_{q}^{O(2)}   = 1 + W_{\text{det}} + V_{2q}^{O(2)} .
\end{equation}
Here, $W_{\text{det}}$ is the Wilson line in the det representation: the appearance of a Wilson line in the fusion of two purely magnetic 't Hooft lines is a novelty not encountered in the compact examples considered previously.

The topological dimension-2 surfaces that link with the 't Hooft lines are similarly in 1-1 correspondence with the topological electric surfaces studied above: there are topological surfaces of quantum dimension 2 for all $\eta \in (0, \pi)$, which can be written as
\begin{equation} 
\tilde T_{\mathrm{e}^{i \eta}}^{O(2)}(\mathcal{M}^{(2)}) = \tilde U_{\mathrm{e}^{i \eta}}^{U(1)}(\mathcal{M}^{(2)}) + \tilde U_{\mathrm{e}^{-i \eta}}^{U(1)}(\mathcal{M}^{(2)}) = \exp\left( \frac{i \theta}{2 \pi} \oint_{\mathcal{M}}  F\right) + \exp\left( -\frac{i \theta}{2 \pi} \oint_{\mathcal{M}}  F \right).
\end{equation}
All of these operators are topological in pure $O(2)$ gauge theory, and they have quantum dimension 2. There is also a surface with $\eta = \pi$, which is given simply by the $\eta = \pi$ magnetic surface of the $U(1)$ gauge theory:
\begin{equation} 
\tilde T_{\mathrm{e}^{i \pi}}^{O(2)}(\mathcal{M}^{(d-2)}) = \tilde U_{\mathrm{e}^{i \pi}}^{U(1)}(\mathcal{M}^{(d-2)}) = \exp\left( \frac{i }{2} \oint_{\mathcal{M}}  F\right) .
\end{equation}
This surface is invertible and squares to the identity operator. The fusion of the magnetic surfaces is identical to that of the electric surfaces considered previously.

The linking between a topological surface $\tilde T(\eta)$ and an 't Hooft line $V_q$ is given precisely by $B_q(\theta)$, defined in \eqref{O2linking}. In the pure $O(2)$ gauge theory, none of the 't Hooft lines $V_q$ can end, and all of the magnetic surfaces ${\tilde T}(\theta)$, $\theta \in (0, \pi)$ are topological. In the presence of a monopole of magnetic charge $q$, $V_q$ is endable, and any surfaces $\tilde T(\eta)$ that do not satisfy $\cos(\eta q) = 1$ will cease to be topological. By the same analysis as in the electric case, the absence of topological magnetic surfaces is equivalent to the completeness of the magnetic spectrum.

Although the fusion algebras of the Wilson lines and electric surfaces match those of the 't Hooft lines and magnetic surfaces, there is one subtle difference between the electric and magnetic sides of the story: in particular, 't Hooft lines of odd charge $q$ are pseudo-real, whereas Wilson lines of odd charge $q$ are real. We can see this by realizing $O(2)$ as a Higgsing of $SO(3)$, as described in Section \ref{section_higgsing_su2}; the odd-charge 't Hooft lines of $O(2)$ descend from the topologically nontrivial 't Hooft lines of $SO(3)$, which correspond to pseudoreal, half-integer spin representations of ${}^L SO(3) = SU(2)$ under S-duality.

In fact, the fusion algebra and the reality properties of the 't Hooft lines (along with the det line) exactly match the representations of $\tilde{O}(2) = (U(1) \rtimes \mathbb{Z}_4)/\mathbb{Z}_2$. Likewise, the topological magnetic surfaces $\tilde{T}(\eta)$ are in 1-1 correspondence with topological Gukov-Witten surfaces of $\tilde{O}(2)$ gauge theory. It is thus natural to identify $\tilde{O}(2)$ gauge theory as the S-dual of $O(2)$ gauge theory in four dimensions. We plan to explore S-duality of disconnected gauge groups further in future work.

Even more generally, $O(2)$ gauge theory has dyonic operators carrying both electric charge $n$ and magnetic charge $m$. These are labeled by a pair of integers $(n, m)$, with the identification $(n, m) \sim (-n, -m)$ due to gauging charge conjugation of $U(1)$ gauge theory. These can be understood as the sum of two $U(1)$ dyonic lines of charge $(n, m)$ and $(-n, -m)$, respectively.

Similarly, there are mixed electric-magnetic topological operators labeled by a pair of angles $(\theta, \eta)$, with the identification $(\theta, \eta) \sim (- \theta, \eta)$. These have quantum dimension 2 unless $\theta$ and $\eta$ are both equal to 0 or $\pi$, in which case they have quantum dimension 1.

The linking coefficient between a $(\theta, \eta)$ surface and an $(n, m)$ line is given by
\begin{equation}
B_{(n, m)}(\theta, \eta) = 2 \cos(\theta n + \eta m) \times  \text{size}(\theta, \eta),
\end{equation}
where size($\theta, \eta$) is the quantum dimension of the surface. If the line $(n, m)$ is endable, then any surface with $\cos(\theta n + \eta m) \neq 1$ will not be topological.

More generally, if dyonic lines of charge $(n_1, m_1)$, ..., $(n_l, m_l)$ are endable, then only surface operators $(\theta, \eta)$ satisfying $\cos(n_i \theta + m_i \eta) = 1$ for all $i$ will be topological. These are precisely the conditions for such an operator to be topological in the $U(1)$ gauge theory, as can be seen from \eqref{U1condition}. Thus, from our discussion of $U(1)$ gauge theory, we see that topological surfaces will exist if and only if some line operators are not endable.

\subsection{The Structure of a Compact Lie Group}\label{ssec:Structure}

The simplest examples of disconnected, compact Lie group are finite groups. One way to obtain disconnected gauge groups that are not finite is to to gauge an outer automorphism of a connected Lie group, as in the $O(2)$ example we have just discussed. Another example is the group $E_8\times E_8$, which has an outer automorphism that acts by swapping the two $E_8$ factors. Gauging this outer automorphism produces a disconnected, compact Lie group $(E_8\times E_8)\rtimes \mathbb{Z}_2$, which is actually the gauge symmetry of what is usually referred  to as the $E_8 \times E_8$ heterotic string theory \cite{Dine:1992ya}. 

In fact, the most general compact Lie group is not that far away from the above construction~\cite{378257,378141}:
\begin{thm}
Let $G$ be a compact (not necessarily connected) Lie group $G$, and let $G_0$ be its identity component. Then 
\begin{equation} G \cong \frac{G_0\rtimes R}{P},\label{g03}\end{equation}
where $R$ is a finite group whose elements act on $G_0$ either trivially or via an outer automorphism of $G_0$, and $P$ is a finite, common normal subgroup of $Z(G_0)$ and $R$.\footnote{Although $P$ is abelian, in general $P$ is \emph{not} central in $G$ because $R$ can act non-trivially on $P \subseteq Z(G_0)$.}
\label{thmstg}
\end{thm}
\begin{proof}
For compact connected $G_0$, the short exact sequence of groups 
 \begin{equation}
\begin{tikzcd}1\arrow{r}&\Inn(G_0)\arrow{r}&\Aut(G_0)\arrow{r}&\Out(G_0)\arrow{r}&1\end{tikzcd}\end{equation}
 splits~\cite{378220}. Since the sequence splits, we can choose a subgroup $R_0\subseteq \text{Aut}(G_0)$ such that $\Aut{G_0} = \Inn{G_0} \rtimes R_0$ (so $R_0 \cong \Out{G_0}$). Now consider the conjugation map $f:G\rightarrow \text{Aut}(G_0)$ (i.e., $f(g)$ maps $h$ to $g^{-1} h g$). The preimage $K=f^{-1}(R_0)$ is a subgroup of $G$ whose intersection with $G_0$ is its center $Z(G_0)$, because all  other elements of $G_0$ map to nontrival elements of $\text{Inn}(G_0)$. 
 
Moreover, $K$ intersects every connected component of $G$. To see this, note that multiplication by an element of $G_0$ maps to composition with an element of $\Inn(G_0)$ under the conjugation map $f$, whereas every element of $\Aut(G_0)$ is the composition of an element of $R_0$ with an element of $\Inn(G_0)$ (since $\Aut{G_0} = \Inn{G_0} \rtimes R_0$).
Putting these two facts together, we conclude that for any $g \in G$ there is some $g_0 \in G_0$ such that the action by conjugation of $g' = g g_0$ lies in $R_0$. Since $g$ and $g' \in K$ lie in the same connected component, $K$ intersects every connected component.

As the preimage of a closed Lie group $R_0$ by a continuous map, the group $K$ is a closed Lie subgroup of the compact Lie group $G$, and hence is a compact Lie group itself. In particular, $K$ therefore has finitely many connected components. As proven in~\cite{150949}, any such group has a finite subgroup $R$ that intersects every connected component. Since $K$ likewise intersects every connected component of $G$, $R$ is a finite subgroup of $G$ that intersects every connected component and whose elements act on $G_0$ either trivially or as an outer automorphism.

Let $P = R\cap G_0$ be the intersection of $R$ with the identity component $G_0$. By construction $P \subseteq Z(G_0)$. Moreover, $P$ is normal in $R$ because $G_0$ is normal in $G$.

Now consider $G_0 \rtimes R$ where $R$ acts on $G_0$ as above. There is a homomorphism $(g_0, r) \to g_0 r$ from $G_0 \rtimes R$ to $G$ given by multiplication in $G$. Because $R$ intersects every connected component of $G$, the homomorphism is surjective, with kernel $(p^{-1}, p)$ for $p \in P$. Thus $G \cong \frac{G_0\rtimes R}{P}$ as claimed.
\end{proof}

In physical terms, this structure theorem tells us that the most general gauge group can be obtained from a connected gauge group $G_0$ by gauging a discrete group $R$ (some of whose elements may act by non-trivial outer automorphisms on $G_0$) and then quotienting by a common subgroup of the discrete group $R$ and the center of the connected group $Z(G_0)$. Note that when $P$ is central in $G_0 \rtimes R$, the quotient by $P$ corresponds to gauging a 1-form symmetry~\cite{Gaiotto:2014kfa}.

\subsection{General Story}\label{sec:General}  

\subsubsection{Characterizing Topological Gukov-Witten Operators in Pure Gauge Theory}
\label{sec:generalTopGW}

As discussed in Section~\ref{sec:Top}, in pure $G$ gauge theory, the topological Gukov-Witten operators $T_{[g]}$ are those labeled by conjugacy classes in the centralizer of the connected component of the identity, $Z_G(G_0)$. It is possible to give a more concrete description, in terms of the topological surfaces we encountered in previous sections on compact gauge theories and finite gauge theories. In particular, following the discussion of Section \ref{ssec:Structure}, let us write $G = (G_0 \rtimes R)/P$, where $G_0$ is compact and every element of $R$ acts on $G_0$ either trivially or by nontrivial outer automorphism. To begin, let us suppose $P$ is trivial, and let $K \subset R$ be the subgroup of $R$ that acts trivially on $G_0$.

Any topological surface operator $T_{[g]}$ must satisfy $\chi_{\text{adj}}([g]) =  \chi_{\text{adj}}(1)$ (we will argue in Section \ref{ssec:Coulomb} below that this is also a sufficient condition). As discussed in Lemma~\ref{lemma1} and Statement~\ref{stat:topGW}, this is equivalent to the statement that $g g_0 g^{-1} = g_0$ for all $g_0$ in $G_0$, the component of $G$ connected to the identity. It is not hard to see that the elements $g := (x, h)$, $x \in G_0$, $h \in R$ satisfying this equation are precisely those of the form $(z, k)$, where $z \in Z(G)$ and $k \in K$. Thus, topological surfaces $T_{[g]}$ are associated with conjugacy classes consisting of elements of this form.

In general, some of these conjugacy classes will be composed of one or more elements of the form $g= (z, 1)$, $1 \neq z \in Z(G)$. Such a surface can be thought of as a gauge-invariant sum of 1-form center symmetry generators of $G_0$, invariant under the action of $R$. For instance, the surfaces $T(\theta), \theta \in [0, \pi]$ of $G=O(2)$ gauge theory that we saw above all take this form---they can be written as gauge invariant sums of $G_0 = U(1)$ 1-form symmetry generators, invariant under the action of $R = \mathbb{Z}_2$.

Some conjugacy classes will be composed of one or more elements of the form $g= (1, k)$, $1 \neq k \in K$. The corresponding topological surfaces are the Gukov-Witten operators of the discrete gauge theory $R$ that link trivially with the endable Wilson line in the adjoint representation.

Finally, some conjugacy classes will consist of elements of the form $g= (z, k)$, $z, k \neq 1$, $z \in Z(G)$, $k \in K$. The associated topological surfaces can be constructed by fusing the above two types of topological surfaces---namely, by fusing one surface $T_{[g]}$, $g = (z, 1)$ with another surface $T_{[g']}$, $g' = (1, k)$.

Thus, we see that the topological Gukov-Witten operators in $G = G_0 \rtimes R$ gauge theory can be written as (a) gauge-invariant sums of 1-form center symmetry generators of $G_0$ gauge theory, (b) Gukov-Witten operators of $R$ gauge theory, or (c) fusions of (a) and (b). This is still true for the most general case of $G = (G_0 \rtimes R)/P$: the quotient by $P$ will simply project out some of these topological Gukov-Witten operators.

\subsubsection{Invertible Gukov-Witten Operators}

As discussed in Section~\ref{sec:generalTop}, not all topological operators are invertible. A topological operator $T$ is invertible if and only if there exists another topological operator $T^{-1}$ such that $T \times T^{-1} = 1$, the trivial operator. By (\ref{eq:qdimfuse}), we see that $T$ may be invertible only if it has quantum dimension 1. As we have seen, topological Gukov-Witten operators, under which Wilson lines are charged, are labeled by conjugacy classes, and the quantum dimension of these surfaces is given by the number of elements in the conjugacy class. This means that an invertible topological surface corresponds to a conjugacy class with a single element, $g$, so that $h g h^{-1} = g$ for all $h\in G$. This is precisely the condition that $g$ lies in the center $Z(G)$ of $G$. Hence, Gukov-Witten operators labeled by elements of the center $Z(G)$ are indeed topological and invertible.

As discussed in Section~\ref{sec:topGW}, the Lemma \ref{lemma1} applied to the adjoint representation shows that the topological Gukov-Witten operators in pure $G$ gauge theory correspond to the centralizer $Z_G(G_0)$, which contains the center $Z(G)$. 
The fact that $Z(G)$ is typically a \emph{proper} subset of $Z_G(G_0)$ for a disconnected gauge group $G$ implies that there will be additional codimension-2 topological operators that are not invertible.

\subsubsection{Electric Completeness and Topological Surfaces}
\label{general_electric_completeness_section}

Next, let us turn to the relationship between general topological operators and completeness of the spectrum. In our language, completeness is equivalent to the statement that every Wilson line $W_\rho$ should be endable. We want to show that if this is not satisfied, then some surface operator $T_{[g]}$ will remain topological, whereas if it is satisfied, then all of these surfaces will be rendered non-topological. From what we have said so far, this is equivalent to the following theorem about group theory:
\begin{thm}
Let $G$ be a compact (not necessarily connected) Lie group, and let $R_{\text{endable}}$ be a subset of the set $R$ of all irreps of $G$ that is closed under tensor product (i.e., if $\rho, \sigma \in \Rend$ with $\rho \otimes \sigma = \bigoplus_i \mu_i$, then $\mu_i \in \Rend$). Then, $\Rend$ is a proper subset of $R$ if and only if there exists a non-trivial element $g \in G$ such that $\chi_\rho(g) = \chi_\rho(1)$ for all $\rho \in \Rend$.
\label{thm2}
\end{thm}
\begin{proof}
First, suppose that such a $g$ exists. By Lemma \ref{lemma1}, $\rho(g) = I$ for all $\rho \in \Rend$. This implies that no representation that decomposes into irreps in $\Rend$ could be faithful, since $g$ acts trivially in every such representation. But, every compact Lie group has a faithful representation, which must include an irrep that is not included in $\Rend$.

Conversely, suppose no such $g$ exists. In this case, the Hilbert space direct sum,
\beq
\rho_{\text{end}} = \bigoplus_{\rho \in R_{\text{endable}}} \rho,
\eeq
is a faithful, unitary representation of $G$, using that every irrep of a compact Lie group is unitary. By Theorem A.10 of \cite{Harlow:2018tng}, any faithful, unitary representation of a compact Lie group has a finite dimensional, faithful sup-rep $\tilde{\rho}_{\text{end}}$.\footnote{In fact, it is possible to find such a finite dimensional faithful representation algorithmically, by ordering all irreps and taking direct sums until a faithful rep is produced. By compactness, this process must terminate after a finite number of steps.} Finally, by Theorem A.11 of \cite{Harlow:2018tng}, the tensor powers of any faithful, unitary representation (and its conjugate) of a compact Lie group generate the entire representation ring, i.e., given $\tilde \rho_{\text{end}}$ finite dimensional and faithful and $\sigma$ an irreducible representation of $G$, there exist $n, m$ such that $\rho_{\text{end}}^{\otimes n} \oplus \bar \rho_{\text{end}}^{\otimes m} = \sigma \oplus ...$. Since $\Rend$ is assumed to be closed under tensor products, we learn that $\Rend$ contains every irreducible representation of $G$, i.e., $\Rend = R$.
\end{proof}
Theorem \ref{thm2} is one of the primary results of our paper. It shows that for a general compact gauge group $G$, if all Wilson lines are endable, then each of the codimension-2 surfaces $T_{[g]}$ is rendered non-topological. Conversely, if some Wilson lines are not endable, then there exists a topological surface $T_{[g]}$, which links trivially with all endable Wilson lines, yet nontrivially with a non-endable Wilson line. This establishes Statement~\ref{stat:ElectricCompleteness}, as promised in  Section~\ref{sec:intro}.

\subsubsection{Twist Vortex Completeness}\label{sssec:twistvortexcomp}

We have already established, by Lemma~\ref{lemmafinite}, that in a theory of gauge fields coupled to dynamical twist vortices, there are no topological Wilson lines if and only if all Gukov-Witten operators are endable. For finite-group gauge theories, the endability of every Gukov-Witten operator necessarily requires the addition of dynamical twist vortices to the theory, labeled by a collection of conjugacy classes whose fusion generates all conjugacy classes in $G$. In fact, even when the identity component $G_0$ of a group is nontrivial, the endability of every Gukov-Witten operator only requires adding additional dynamical twist vortices whose corresponding conjugacy classes generate every conjugacy class in the finite group $\pi_0(G)$ under the quotient map $G \to \pi_0(G)$.

In order to see this, recall that the Gukov-Witten operators corresponding to the conjugacy classes in $G_0$ are already endable in pure gauge theory on improperly quantized 't Hooft operators, by Statement~\ref{stat:endableGW}. Suppose that we add to the theory dynamical twist vortices labeled by some conjugacy class $[g]$ in $G$ that intersects a non-identity connected component $F$. By fusion with the vortices labeled by conjugacy classes in $G_0$, we can generate dynamical vortices labeled by the conjugacy classes of any group element in $F$, since the action of $G_0$ on $F$ by left or right multiplication is transitive. In order to generate every conjugacy class in $G$, it is thus sufficient to add a collection of twist vortices that generate conjugacy classes intersecting every connected component of $G$, or equivalently, which generate all conjugacy classes in $\pi_0(G)$, as claimed.

Thus, even for disconnected Lie groups $G$ with a nontrivial identity component $G_0$, the collection of additional, dynamical twist vortices is still controlled by a finite group, namely $\pi_0(G)$. This establishes Statement~\ref{stat:TwistVortexCompleteness} in  Section~\ref{sec:intro}. In 4d theories, these twist vortices can provide examples of ``cosmic strings,'' with potential phenomenological implications that we will sketch in Section~\ref{sec:cosmicstrings}.

\subsubsection{Magnetic Completeness and Global Symmetry}

So far, we have focused primarily on Wilson lines and the electric side of the story. In this section, we briefly comment on 't Hooft operators and magnetic global symmetries.
To begin, we restrict ourselves to the study of $G$ gauge theory with $G=(G_0 \times H)/P$, $G_0$ connected, $H$ finite, in $d \geq 3$ dimensions. Note in particular that the semidirect product between $G_0$ and $H$ studied above has been replaced by a direct product.

In such a theory, 't Hooft operators may be classified by nontrivial $G$ gauge bundles. These gauge bundles are labeled by elements of $\pi_1(G_0)$, the first fundamental group of the identity component of $G$. As in the case of a connected gauge group discussed in Section \ref{sec:connectedmagnetic}, there is a magnetic $(d-3)$-form $\pi_1(G_0)^\vee$ global symmetry, where $\pi_1(G_0)^\vee$ is the Pontryagin dual group of $\pi_1(G_0)$.

't Hooft operators can end in the presence of magnetically charged $(d-4)$-branes (i.e., monopoles). The set of endable 't Hooft operators forms a subgroup $N$ of $\pi_1(G_0)$, which is normal since $\pi_1(G_0)$ is abelian. The remaining magnetic global symmetry is given by $({\pi}_1(G_0)/N)^\vee$, the Pontryagin dual of ${\pi}_1(G_0)/N$. From this, we see that if all 't Hooft operators are endable, then the magnetic symmetry will be completely broken. If not, then some nontrivial remnant will exist. Notably, all of the topological magnetic operators of dimension $2$ are invertible.

How is this story modified in the more general case $G=(G_0 \rtimes H)/P$, where some elements of $H$ acts via nontrivial outer automorphism on $G_0$? First, the set of 't Hooft operators is now classified by $\pi_1(G)/\pi_0(G)$, where elements of $\pi_0(G)$ act via conjugation on loops in $\pi_1(G)$ with basepoint equal to the identity. The action of $\pi_0(G)$ will in general identify distinct elements of $\pi_1(G)$, and as a result $\pi_1(G)/\pi_0(G)$ is not generically a group, as we saw in the case of $O(2)$ gauge theory.

The category of topological magnetic operators is morally equal to the ``Pontryagin dual'' of $\pi_1(G)/\pi_0(G)$. However, when $\pi_1(G)/\pi_0(G)$ is not a group, its Pontryagin dual is not well-defined. From the $O(2)$ example above, we expect that in this case some of the topological magnetic operators will be non-invertible. One way to understand this category would be to map it to the category of topological Gukov-Witten operators of some other gauge group $^LG$, related to $G$ via S-duality. For instance, for $G=O(2)$, we saw that the spectrum of 't Hooft lines and topological magnetic surfaces in 4d could be understood as the spectrum of Wilson lines and topological Gukov-Witten surfaces of $\tilde{O}(2)$ gauge theory. We leave further exploration of magnetic completeness, topological operators, and S-duality of these disconnected gauge groups to future work.

\section{Noncompact Gauge Groups}\label{sec:Noncompact}

Our previous results regarding the correspondence between topological (endable) Gukov-Witten operators and endable (topological) Wilson lines does not work for noncompact gauge groups. In this section we illustrate the differences in a couple of examples.

\subsection{\texorpdfstring{$\mathbb{R}$}{R} Gauge Theory}

We used the compactness of $G$ several times in the proof of Theorem \ref{thm2}. It is interesting to consider see how this theorem is violated if $G$ is not compact. In particular, consider $G =\mathbb{R}$, and suppose that the Wilson lines of charge $q = m + n \sqrt{2}$ are endable for $m$, $n \in \mathbb{Z}$. This spectrum is incomplete---for instance, a Wilson line of charge $1/2$ cannot end. However, no surfaces will remain topological: the characters in question are given by $\chi_q(\alpha) := \exp(2 \pi i q \alpha)$, and for any $\alpha \in \mathbb{R}$, there exists $q$ such that the Wilson line of charge $q$ is endable, and $\chi_q(\alpha) \neq 1$. 

It is interesting to note, however, that although the spectrum is not complete in this example, the charges of the endable Wilson form a dense subset of the real numbers. Indeed, for $\mathbb{R}$ gauge theory, a slightly modified version of Theorem \ref{thm2} holds true: there exist topological surfaces if and only if the set of charges of endable Wilson lines is not dense in $\mathbb{R}$. One direction of this statement is obvious: if this set is dense, then no $\alpha \in \mathbb{R}$ will satisfy $\chi_q(\alpha) := \exp(2 \pi i q \alpha) = 1$ for all endable Wilson lines $W_q$. Conversely, if the set of charges of endable Wilson lines is not dense, then there exists a neighborhood $(- \delta , \delta )$ of the origin that does not contain any endable Wilson line $W_q$, $q \neq 0$. Allowing $\delta$ to be the value of the largest such neighborhood, we conclude that $\delta$ must in fact be the smallest positive charge of an endable Wilson line. From here, we see that $W_{\pm \delta}, W_{\pm 2 \delta}, W_{\pm 3 \delta}, ...$. are all endable, and indeed these must be the only endable Wilson lines: given some other endable Wilson line $W_\gamma$, then if $\gamma \in (n \delta, (n+1) \delta)$, then $W_\beta$ with $\beta = |\gamma - n \delta| < \delta$ is endable, contradicting our assumption that $\delta$ was the smallest positive charge of an endable Wilson line. As a result, surfaces with $\alpha = n/\delta$ will be topological for all $n \in \mathbb{Z}$, as $\chi_{m \delta}(n \alpha) := \exp(2 \pi i m n)  = 1$.

\subsection{\texorpdfstring{$\mathbb{Z}$}{Z} Gauge Theory}

We next consider pure $\mathbb{Z}$ gauge theory in $d$ dimensions. Wilson lines are labeled by phases $\alpha \in [0, 2 \pi)$, whereas Gukov-Witten operators are labeled by integers. All of these operators are topological: the Wilson lines generate a $U(1)$ $(d-2)$-form global symmetry, whereas the Gukov-Witten operators generate a $\mathbb{Z}$ 1-form global symmetry. The charge operators of the 1-form symmetry are charged under the $(d-2)$-form symmetry, and vice versa (as is typical for a discrete gauge symmetry).

In the presence of charge $\alpha$ matter, the Wilson line of charge $\alpha$ is endable. The Gukov-Witten operator labeled by the integer $n$ is not topological unless $\alpha n \in 2 \pi \mathbb{Z}$. If $\alpha / \pi$ is rational, therefore, a $\mathbb{Z}$ 1-form global symmetry is preserved. If $\alpha / \pi$ is not rational, however, then the 1-form symmetry is completely broken, even though the spectrum is incomplete. As in the $\mathbb{R}$ gauge theory example considered above, however, the set of endable Wilson lines in this case form a dense subset of $[0, 2\pi)$. On the other hand, if the Gukov-Witten operator labeled by $n$ is endable, then the $(d-2)$-form global symmetry is broken to $\mathbb{Z}_n$, so the spectrum of Gukov-Witten operators is complete precisely when the $(d-2)$-form global symmetry is absent.

Unbroken $\mathbb{Z}$ gauge theory is rarely considered in the study of quantum field theory and quantum gravity. However, \emph{spontaneously broken} $\mathbb{Z}$ gauge groups appear in rather familiar examples of field theories, and the results are quite different then. Let us consider the theory of a compact scalar field $\theta$ in four dimensions,
\begin{equation}
\mathcal{L} = \int \rmd \theta \wedge \star \rmd \theta,
\end{equation}
with $\theta := \theta + 2 \pi$. This theory has a $U(1)$ $0$-form global symmetry given by $\theta \rightarrow \theta + c$ for $c \in [0, 2 \pi)$. Additionally, the compact scalar can be dualized to a 2-form $B_2$, and there is a 2-form magnetic global symmetry given by shifting $B_2$ by a flat connection, $B_2 \rightarrow B_2 + C_2$, $\rmd C_2 = 0$.

Alternatively, we can view this system as the theory of a noncompact real scalar $\phi$ with a spontaneously broken gauge symmetry $\mathbb{Z}$, under which $\phi$ transforms as $\phi \rightarrow \phi + 2\pi$. Before gauging this $\mathbb{Z}$ 0-form symmetry, there is a 0-form global symmetry $\mathbb{R}$. Such a symmetry is associated with topological codimension-1 surfaces,
\begin{equation}
U_{\mathrm{e}^{i\alpha}}({M_3}) = \mathrm{e}^{i \alpha \oint_{M_3} \star \rmd \phi},
\end{equation}
which act on local operators $\mathrm{e}^{i \beta \phi(x)}$ by 
\begin{equation}
U_{\mathrm{e}^{i\alpha}}(S^3(x))  \mathrm{e}^{i \beta \phi(x)} = \mathrm{e}^{i \alpha \beta} \mathrm{e}^{i \beta \phi(x)},
\end{equation}
for $\alpha, \beta \in \mathbb{R}$, and $S^3(x)$ a 3-sphere whose interior contains the point $x$.

Gauging the $\mathbb{Z}$ symmetry amounts to restricting $\beta$ to lie in $\mathbb{Z}$ and imposes the condition that the periods of $\star \rmd \phi$ must be integers, which identifies $\alpha \sim \alpha + 2\pi$. As a result, the $\mathbb{R}$ 0-form global symmetry is reduced to the quotient $\mathbb{R} / \mathbb{Z} \simeq U(1)$.

As discussed above, pure $\mathbb{Z}$ gauge symmetry has two remaining symmetries: a $\mathbb{Z}$ electric 1-form symmetry, and a $U(1)$ magnetic 2-form symmetry generated by the Wilson lines of the gauge theory. 
However, in the case at hand, the 0-form gauge symmetry $\mathbb{Z}$ is spontaneously broken by the field $\phi$, and the Wilson line in $\mathbb{Z}$ representation $\alpha \in [0, 2 \pi)$, namely $\mathrm{e}^{i \alpha \oint \rmd \phi}$, can end on an operator of the form $\mathrm{e}^{i \alpha \phi}$. This means that the magnetic surfaces labeled by integers are not topological, so there is no $\mathbb{Z}$ 1-form electric global symmetry that results from gauging the axion shift symmetry. Instead, there is only a $U(1)$ 2-form global symmetry generated by the Wilson lines, under which the magnetic surfaces are charged.

Thus, we are left with a 0-form $U(1)$ global symmetry and a 2-form $U(1)$ global symmetry, exactly as we had for the compact scalar field $\theta$, with associated 3/1-form currents $\star \rmd \phi$ and $\rmd \phi$, respectively. The symmetry generators for the 2-form symmetry $\exp(i \alpha \int_\gamma \rmd \phi )$ are the Wilson lines of the $\mathbb{Z}$ gauge theory.
Since these groups are now compact, we can use Statements \eqref{stat:ElectricCompleteness} and \eqref{stat:TwistVortexCompleteness} again to recover the correspondence between absence of topological operators and completeness of the spectrum. It is worth noticing that whenever we find these noncompact gauge groups in string theory, they are indeed spontaneously broken to a compact group, so this is consistent with the expectation of no global symmetries and the Completeness Hypothesis in quantum gravity. The canonical example is SL(2,$\mathbb{Z}$) in Type IIB, which gets spontaneously broken to finite subgroups in the moduli space.

\section{Higgsing}\label{sec:Higgsing}

In the previous sections we have shown that, for compact gauge groups, completeness of the gauge theory spectrum is equivalent to the absence of topological Gukov-Witten operators, while completeness of twist vortices is equivalent to the absence of topological Wilson lines. One could wonder how robust is this story under Higgsing. In other words, could the process of Higgsing change the spectrum of topological operators such that a complete UV theory gives rise to an incomplete IR theory, or vice versa? Here, we show that the process of Higgsing preserves completeness: the IR theory is complete if and only if the UV theory is also complete. We further discuss the behavior of Gukov-Witten operators and 't Hooft operators under Higgsing, and we argue that any Lie group can be obtained as a result of Higgsing an $SU(n)$ gauge theory for sufficiently large $n$. The resulting theory will have no topological Wilson lines, and it will have no magnetic $(d-3)$-form symmetry.

We begin with a number of examples.

\subsection{Higgsing \texorpdfstring{$SU(2)$}{SU2} and \texorpdfstring{$SO(3)$}{SO3}}\label{section_higgsing_su2}

We start with Higgsings of pure $SU(2)$ gauge theory (with zero theta angle).  The $SU(2)$ theory has Wilson lines labeled by $SU(2)$ representations. Adjoint lines (and more generally, integer spin representations) are endable, and so the only non-endable Wilson lines are those that represent the $\mathbb{Z}_2$ center nontrivially (half-integer spin). None of these lines are topological.

The Gukov-Witten operators correspond to conjugacy classes of $SU(2)$. Recall from Section \ref{sec:SUNZK} that these conjugacy classes are labeled by $\theta \in [0, \pi]$, where a representative of the conjugacy class $\theta$ is given by diag$(\exp (i \theta), \exp( - i \theta))$.
The center of $SU(2)$ is isomorphic to $\mathbb{Z}_2$ and consists of the elements $\theta = 0 , \pi$. The Gukov-Witten operator associated to $\theta =0$ is the trivial operator, and the $\theta = \pi$ Gukov-Witten operator is a topological, invertible operator which generates the 1-form center symmetry. Gukov-Witten operators corresponding to $\theta \in ( 0 , \pi)$, on the other hand, are not topological. All of these Gukov-Witten operators are endable, and relatedly there are no nontrivial topological Wilson lines.
There are also no topologically nontrivial 't Hooft operators, since there are no nontrivial $SU(2)$ bundles over $S^2$ (see Section \ref{sec:tHO}).

Having discussed the operators in the $SU(2)$ theory, we now describe how they map to different operators after Higgsing. 

\paragraph{Higgsing to $U(1)$:} In this textbook example, we turn on a vev for a scalar $\Phi$ in the adjoint representation to break the gauge group to $U(1)$ in the Cartan. The $SU(2)$ Wilson line of spin $j$ becomes a sum of  $U(1)$ Wilson lines of charge $2s$ for $s= -j, -j+1,\cdots, j-1,j$, where we have normalized the IR quantum of charge to be 1. As before, Wilson lines of odd charge are topological. The matching of Gukov-Witten operators is subtler. Gukov-Witten operators in the IR theory are labeled by an angle $\theta \in [0, 2 \pi)$. In pure $U(1)$ gauge theory, these are all topological and generate a $U(1)$ 1-form symmetry, but in the case at hand this symmetry is broken to a $\mathbb{Z}_2$ subgroup by the Higgs field. 

The resulting $\mathbb{Z}_2$ topological operator is the IR description of the UV topological Gukov-Witten operator. Naively, we seem to have twice as many non-topological IR Gukov-Witten operators (which are labeled by an angle $\theta \in (0, 2\pi)$ other than $\pi$) as non-topological UV Gukov-Witten operators (which are labeled by an angle in $(0,\pi)$. To understand this mismatch, notice that the UV Gukov-Witten operator was defined by first excising a codimension-2 surface from the path integral and then specifying the holonomy in the transverse $S^1$ to be $g\in SU(2)$, and then summing over all $g$ in the conjugacy class. When attempting to do this in the spontaneously broken vacuum,  any $g$ that does not preserve the adjoint Higgs vev $\Phi$ will lead to a discontinuity and disappear from the path integral. Only those $g\in [g]$ that commute with $\Phi$ will contribute. In our case, writing the Higgs vev as 
\begin{equation}  \Phi= \phi \left(\begin{array}{cc} 1&0\\0&-1\end{array}\right),\quad \phi>0,\label{wee3}\end{equation} 
we get that for an $SU(2)$  conjugacy class with parameter $\theta$, the sum collapses\footnote{There are some issues of normalization and dividing by volume of the gauge group.} to the elements
\begin{equation} g= \left(\begin{array}{cc} \mathrm{e}^{i\theta}&0\\0&\mathrm{e}^{-i\theta}\end{array}\right)\quad\text{and}\quad g'= \left(\begin{array}{cc} \mathrm{e}^{-i\theta}&0\\0&\mathrm{e}^{i\theta}\end{array}\right).\label{ewe334}\end{equation}
These two elements belong to the same $SU(2)$ conjugacy class, but they are in different conjugacy classes in $U(1)$; thus, the UV Gukov-Witten operator with parameter $\theta$ becomes a direct sum of two IR Gukov-Witten operators. This explains the mismatch noted above. 

More general linear combinations of the IR Gukov-Witten operators can be obtained from the UV theory as operators with insertions of the Higgs field $\Phi$. For example, consider the same definition for the Gukov-Witten operator, as a sum over boundary conditions on $S^1$, but introduce additional $\Phi$-dependent factors in the sum, such as
\begin{equation}
\sum_{g\in [g]}\quad\rightarrow\quad \sum_{g\in [g]} {\text{tr}(g\cdot\Phi)}. \label{gwpres}
\end{equation}
In the broken phase, with the vev \eqref{wee3}, again only the two elements \eqref{ewe334} contribute, but the additional Higgs-dependent factors ensure that we produce the linear combination
\begin{equation} 
2i\phi\sin\theta \left(U_{g=\mathrm{e}^{i\theta}} - U_{g=\mathrm{e}^{-i\theta}}\right)\end{equation}
of the IR theory. In this way, all surface operators of the IR theory originate from operators in the UV theory dressed with Higgs fields. We emphasize that the particular dressing \eqref{gwpres} is only one arbitrary example, not of special physical relevance.

Finally, the IR theory 't Hooft operators are labeled by $\pi_1(U(1))=\mathbb{Z}$. In contrast, $\pi_1(SU(2))$ is trivial, so there are no nontrivial 't Hooft operators in the UV theory. Again, we can construct (many) UV operators involving the Higgs field that become the IR 't Hooft operator. For instance, consider the operator defined by excising a line and prescribing boundary conditions on the angular $S^2$ as
\begin{equation}\Phi \rightarrow \Phi_0 \cdot f_{S^2},\label{eww}\end{equation}
where  $f_{S^2}$ is a representative of the nontrivial map from $S^2$ to $SU(2)/U(1)=S^2$. In the IR, this behaves as an 't Hooft operator in the spontaneously broken vacuum for any nonzero value of $\Phi_0$.

Since there are no non-endable 't Hooft operators in the UV, the 't Hooft operators we just constructed must be endable. As is well-known, the IR theory contains 't Hooft/Polyakov monopole solutions \cite{Polyakov:1974ek,tHooft:1974kcl}, and the 't Hooft operators can end on operators that create these monopoles. In terms of the UV description \eqref{eww}, the line can end by simply deforming $\Phi_0$ to zero.

\paragraph{Higgsing to $\tilde{O}(2)$ and $O(2)$:} Consider the spin 2 representation of $SU(2)$. Regarding $SU(2)$ as the double cover of $SO(3)$, the spin 2 representation corresponds to a symmetric tensor. If we choose coordinates $(x, y, z)$ for the vector representation, then a vev like 
\begin{equation} \langle \Phi \rangle= \mathrm{d}z^2\label{vev0}\end{equation}
only preserves rotations in the $xy$ plane, but it also preserves simultaneous reflections of the $x$ and $z$ or $y$ and $z$ coordinates. The unbroken subgroup will be a double cover of $O(2)$, and in particular, since these simultaneous reflections are equivalent to $180^\circ$ rotations in the $xz$ or $yz$ planes, the reflections will square to $-1$ in the double cover. The resulting group is $\text{Pin}^-(2)$, or equivalently,
\begin{equation}\text{Pin}^-(2)=\tilde{O}(2)=\frac{U(1)\rtimes \mathbb{Z}_4}{\mathbb{Z}_2}.\end{equation}
We have introduced the notation $\tilde{O}(n)$ to denote generally the extension of $O(n)$ by reflections that square to $(-1)$,
\begin{equation} \tilde{O}(n)\equiv \frac{SO(n)\rtimes \mathbb{Z}_4}{\mathbb{Z}_2}.\end{equation}
Note that $\text{Pin}^-(2)=\tilde{O}(2)$, but this is a low-dimensional accidental isomorphism.

Going back to $\tilde{O}(2)$, we can now match Wilson lines and Gukov-Witten operators.
The group $\tilde{O}(2)$ can also be described explicitly as the group of $2 \times 2$ matrices generated by 
\begin{equation} M_{\theta}=\left(\begin{array}{cc}\cos(\theta)&-\sin(\theta)\\\sin(\theta)&\cos(\theta)\end{array}\right),\quad R= i\left(\begin{array}{cc}0&1\\-1&0\end{array}\right).\label{em4}\end{equation}
Notice that this is essentially the same matrix description of $O(2)$, except that reflections are multiplied by $i$.
The representation theory of $O(2)$ and $\tilde O(2)$ are very similar, including the adjoint (which represents $R$ by a sign) and two-dimensional representations indexed by an integer $q$. Unlike for $O(2)$, only representations of even $q$ are real; those of odd $q$ are pseudoreal instead (these are the spinor representations when we regard  $\tilde{O}(2)$ as $\text{Pin}^-(2)$).

 The adjoint of $SU(2)$ decomposes as a sum of the sign representation and the charge $q=2$ real representation. This means that in the IR theory, the adjoint (sign) representation and all the real representations of even charge are always endable. The fundamental representation of $SU(2)$ becomes the pseudoreal representation of charge $q=1$. 

The conjugacy classes of $\tilde{O}(2)$ are similar to those of $O(2)$, discussed in Section \ref{sec:O2}: there is one conjugacy class for each $\theta\in [0,\pi]$, as well as a conjugacy class for reflections. The relation between $SU(2)$ and $\tilde{O}(2)$ conjugacy classes is one-to-one for classes of determinant $+1$---they are described by the same set of data. For instance, the pseudoreal $q=1$ line will be endable if the UV $SU(2)$ theory has fields transforming in the fundamental representation. The Gukov-Witten operator associated to reflections is rendered non-topological since it links with the determinant line, which is endable. This Gukov-Witten operator is not present as a Gukov-Witten operator in the UV gauge theory, but only as an operator involving the Higgs field, similar to the discussion around \eqref{gwpres}.

Just like in the $SU(2)\rightarrow U(1)$ example above, we also have 't Hooft operators, labeled by principal $O(2)$ bundles on the transverse $S^2$. These are classified by the equatorial transition function $S^1\rightarrow O(2)$ up to the conjugation action. This can be used to fix a basepoint, so 't Hooft operators are labeled by $\mathbb{Z}$. They can all end on operators that create monopoles, constructed in the same way as before.

In pure $\tilde O(2)$ gauge theory, the Wilson line in the det representation is topological. In the case at hand, this line is no longer topological, however, as the Gukov-Witten operator it links can end on an operator that creates a vortex. These vortices are just ordinary ANO (Abrikosov-Nielsen-Olesen) strings \cite{Abrikosov:1956sx, Nielsen:1973cs} (or their non-abelian generalization \cite{Alford:1992yx}), constructed by allowing the Higgs field $\Phi$ to wind appropriately around the core of the string. This winding corresponds to flipping the direction of the vev \eqref{vev0} as we wind around the string. Specifically, if $\varphi$ is an angular coordinate centered on the string, we have the holonomy
\begin{equation}\langle\Phi\rangle(\varphi)= (\cos(\varphi/2) \mathrm{d}z+\sin(\varphi/2) \mathrm{d}y)^2,\quad 0\leq \varphi\leq 2\pi.\end{equation}
After one turn, the vev remains the same, but the coordinate $z$ has flipped sign, specifying a holonomy by a reflection. 

This concludes our discussion of $\tilde{O}(2)$. We now briefly consider $O(2)$, to connect with Section \ref{sec:O2}. $O(2)$ is not a subgroup of $SU(2)$, since the reflections have determinant $-1$, but it is a subgroup of $U(2)$ or $SO(3)$. We will consider the latter possibility, where breaking to $O(2)$ is achieved by the same vev \eqref{vev0} in a symmetric tensor representation. Since $SO(3)=SU(2)/\mathbb{Z}_2$, there is a surjective map from $SU(2)$ conjugacy classes to $SO(3)$ conjugacy classes, which are labeled by $\theta' = 2 \theta \in [0, \pi]$. This map identifies conjugacy classes with angle $\theta$ and $\pi-\theta$. In particular, the center conjugacy class in $SU(2)$ is mapped to the identity, and the classes with angle $\theta$ and $\pi-\theta$ are mapped to the same conjugacy class, labeled by $\theta' = 2 \theta$. Accordingly, only representations which have the same character for $\theta$ and $\pi-\theta$ survive; as is well-known, these are precisely the integer-spin representations of $SU(2)$. 

The map from $SU(2)$ to $SO(3)$ descends to a map from $\tilde{O}(2)$ to $O(2)$, and a similar story takes place for Gukov-Witten operators and Wilson lines. The only novelty is the presence of an 't Hooft operator in $SO(3)$, associated to an $SO(3)$ bundle with nontrivial Stiefel-Whitney class on the transverse $S^2$. This line survives and becomes a genuine 't Hooft operator of the IR theory. Finally, the topological Wilson line in the det representation is rendered non-topological by the same ANO string as in the $\tilde{O}(2)$ theory.

\subsection{Higgsing \texorpdfstring{$U(1)\rightarrow \mathbb{Z}_N$}{U1toZN}}

Let us move now to the canonical example of Higgsing $U(1)$ to the discrete subgroup $\mathbb{Z}_N$, as both groups have been independently studied in this work. The set of operators of $U(1)$ were discussed in Section \ref{sec:U1}, where we explained that there Wilson lines and 't Hooft operators are labeled by integers $n\in \mathbb{Z}$ and are not topological and not endable (in the absence of charged states). There are also electric and magnetic Gukov-Witten operators given in \eqref{UFe},\eqref{UFm} labeled by a phase $\alpha\in [0,2\pi)$, which are topological and endable in improperly quantized Wilson lines/'t Hooft operators of charge $n=\alpha$ respectively.

By contrast, as discussed in Section \ref{sec:ZN}, the Wilson lines in $\mathbb{Z}_N$ are labeled by integers $n=0, \ldots, N-1$ and are topological but not endable. Furthermore, Gukov-Witten operators of the electric symmetry are also labeled by integers $n=0, \ldots, N-1$, which are topological and not endable. This might seem to contradict the $U(1)$ case, where every Gukov-Witten operator is endable, and no nontrivial Wilson line is topological. The key to resolving this apparent contradiction is to note that the Higgsing produces twist vortices, which allow all electric Gukov-Witten operators to end and, therefore, render the electric Wilson lines non-topological. In four dimensions, these twist vortices are the familiar strings of the Abelian Higgs model, which are charged under the 2-form global symmetry of the $B$-field dual to the axion that is eaten up by the gauge field. The Wilson surfaces charged under the 2-form global symmetry are the same as the electric Gukov-Witten operators (see \eqref{UB}), so these get indeed broken in the presence of the strings. Hence, completeness of the spectrum of twist vortices is guaranteed in the Higgsed theory, and this in turn implies the absence of topological Wilson line operators, as expected.

The 't Hooft operators of the pure $U(1)$ gauge theory are labeled by integers, and they are charged under a $(d-3)$-form $U(1)$ global symmetry. Upon Higgsing to $\mathbb{Z}_N$ and flowing to the deep IR, this 1-form magnetic symmetry is gauged, and as a result there are no 't Hooft operators in the IR theory.

\subsection{Higgsing \texorpdfstring{$O(2)$}{O2}}

\paragraph{Higgsing to Dih$_q$:}
Let us next consider the Higgsing of $O(2)$ gauge theory to a discrete subgroup. In particular, suppose we condense some matter field of charge $q$. We will first consider $q=2$, then $q=3$ and then general $q$.

When there is matter in the $q=2$ representation of $O(2)$, the $U(1)$ part of the gauge group is Higgsed down to $\mathbb{Z}_2$. The remaining gauge group is $\mathbb{Z}_2 \times \mathbb{Z}_2$. In general, $\mathbb{Z}_2 \times \mathbb{Z}_2$ gauge theory has Wilson lines $e_1$ and $e_2$ associated with each $\mathbb{Z}_2$ factor, which fuse to give the line $e_1 e_2$. Similarly, it has magnetic surfaces $m_1$ and $m_2$, which fuse to give $m_1 m_2$. Each $e_i$ links nontrivially with $m_i$, yielding a phase of $-1$. 

In the UV theory, there is just one conjugacy class with a topological surface: namely, the $\theta = \pi$ class, which is associated with the $\mathbb{Z}_2$ center symmetry of $O(2)$. This surface is endable, as such a surface may end on an improperly quantized 't Hooft operator. Similarly, there is one topological Wilson line $W_{\text{det}}$, associated with the det rep. This line links nontrivially with a surface $T^{O(2)}_{\text{disc}}$ in the conjugacy class of $O(2)$ that is disconnected from the identity---this surface is not topological, due to the endability of the det line, nor is it endable.

Thus, we have the following dictionary from the $O(2)$ gauge theory to the $\mathbb{Z}_2 \times \mathbb{Z}_2$ gauge theory: 
\begin{align}
\begin{tabular}{c|c}
O(2) & $\mathbb{Z}_2 \times \mathbb{Z}_2$ \\ \hline
$T^{O(2)}(\pi)$ & $ m_1$ \\
$T^{O(2)}_{\text{disc}} $& $ m_2$ \\ 
$W^{O(2)}_{q=1}$ & $e_1 $\\
$W_\text{det}$ & $e_2$
\end{tabular}
\end{align}
Here, $m_1$ and $e_2$ are endable, while $e_1$, $m_2$, $e_1 e_2$, and $m_1 m_2$ are not endable. As a result, $m_1$ and $e_2$ are topological, while the rest of the lines and surfaces are not. Clearly, the resulting $\mathbb{Z}_2 \times \mathbb{Z}_2$ gauge theory is not pure---some surfaces and lines are endable, so not every surface/line is topological.

Next, we consider the case of Higgsing by a field of charge $q=3$. Now, the $U(1)$ part of the gauge group is Higgsed to $\mathbb{Z}_3$. The remaining gauge group is a nonabelian group of order 6, namely, $S_3$.

$S_3$ gauge theory was discussed previously in Section \ref{ssec:S3}. It features nontrivial Gukov-Witten operators $T_{[\theta]}$, $T_{[\tau]}$ labeled by two nontrivial conjugacy classes, $[\theta]$ and $[\tau]$, of size $2$ and $3$, respectively. It features Wilson lines $W_-$, $W_2$, labeled by the two nontrivial representations of $S_3$, namely the sign representation and the standard representation. These representations have dimensions 1 and 2, respectively.

Upon Higgsing the $O(2)$ gauge theory by a $q=3$ particle, one topological surface remains, namely $T^{O(2)}(\theta=2 \pi/3)$. This surface is also endable, and it descends to the $S_3$ surface $T_{[\theta]}$. The $S_3$ surface $T_{[\tau]}$ comes from surface $T^{O(2)}_{\text{disc}} $ associated with the disconnected component of the identity, and it is neither endable nor topological. The topological, endable line $W_\text{det}$ descends to the line $W_{-}$, whereas the line $W_{2}$ comes from any Wilson line $W^{O(2)}_{q'}$ with $q' \neq 0$ mod 3. Such a line is neither topological nor endable. Thus, we have the dictionary:
\begin{align}
\begin{tabular}{c|c}
O(2) & $S_3$ \\ \hline
$T^{O(2)}(2\pi/3)$ & $ T_{[\theta]}$ \\
$T^{O(2)}_{\text{disc}} $& $ T_{[\tau]}$ \\ 
$W^{O(2)}_{q'}~(q' \neq 0$ mod 3) & $W_2 $\\
$W_\text{det}$ & $W_-$
\end{tabular}
\end{align}
Note that the quantum dimensions of the topological lines and surfaces match up correctly: the size of the $O(2)$ conjugacy class with $\theta = 2 \pi/3$ is two, which is also the size of the conjugacy class $\theta$ of $S_3$. The sign rep and the det rep are both 1-dimensional, whereas the $q'$ rep and the standard rep are both 2-dimensional.

Finally, let us consider the condensation of a scalar in a general charge $q$ representation. The resulting group is the dihedral group with $2q$ elements, Dih$_q$. This group is represented as
\begin{equation}
\text{Dih}_q = \langle x, a | x^2 = a^q = e\,,~~x a x^{-1} = a^{-1}\rangle.
\end{equation}
Here, $a$ should be thought of as the generator of the $\mathbb{Z}_q$ remaining after Higgsing $U(1)$ with a charge $q$ particle, while $x$ represents charge conjugation. The topological surfaces (all of which are endable) are given by $T^{O(2)}(\theta)$ with $\theta = 2 \pi k/q$, $k=1, 2,...,\lfloor q/2 \rfloor$. For $\theta \neq \pi$, these surfaces have quantum dimension 2, and they correspond to conjugacy classes containing the elements $\{ a^k, a^{-k} \}$. If $q$ is even, then the surface $T^{O(2)}(\theta = \pi)$ is an invertible topological surface, and it corresponds to the conjugacy class of the element $a^{q/2}$, which is in the center of Dih$_q$.
 The remaining conjugacy classes of Dih$_q$ all feature elements of the form $[x a^m]$ for some $m$, which means that they come from the disconnected component of $O(2)$: such surfaces are neither endable nor topological.
 
The det Wilson line $W_{\text{det}}$ of $O(2)$ gauge theory descends to det Wilson line of Dih$_q$. This is an endable, topological line corresponding to the representation of Dih$_q$ whose kernel consists of all elements of the form $a^m$.\footnote{The existence of such a representation is guaranteed by the fact that $\{e, a, a^2, ..., a^{q-1} \}$ is a normal subgroup of Dih$_q$, and any normal subgroup is the kernel of some homomorphism. In this case, the homomorphism maps Dih$_q \rightarrow \mathbb{Z}_2$, so the nontrivial irrep of $\mathbb{Z}_2$ induces a nontrivial irrep of Dih$_q$.} The other Wilson lines of Dih$_q$ gauge theory come from Wilson lines of charge $q'$. Here, there is an identification of representations $q' \sim q' + q \sim q - q'$, which means that for $q$ odd, there are two-dimensional representations of Dih$_q$ given by $q' = 1, 2,...,(q-1)/2$, and for $q$ even there are two-dimensional representations given by $q' = 1, 2,...,q/2-1$. For $q$ even, the $q/2$ rep splits into a pair of 1-dimensional representations (as we saw for the case of $q=2$ above, in which the $q=1$ line decomposes into the irreps $e_1$ and $e_1 e_2$). Thus, when combined with the trivial rep and the one-dimensional irrep coming from the det rep, this gives a total of $(q+ 3)/2$ irreps for $q$ odd and $(q+6)/2$ irreps for $q$ even, as expected. Aside from the det Wilson line, none of these Wilson lines are topological or endable.

\paragraph{Higgsing to $U(1)$:}
We may also Higgs $O(2)$ to $U(1)$ by giving a vev to an adjoint-valued scalar field. As in \eqref{O2Wilsonsplitting}, the $O(2)$ Wilson line of charge $q$ descends to a pair of $U(1)$ Wilson lines of charge $q$ and $-q$. The spectrum of the IR $U(1)$ gauge theory is complete if and only if the spectrum of the UV $O(2)$ gauge theory is complete: namely, if there exists matter in the $q=1$ representation of $O(2)$. In four dimensions, an analogous story holds for the 't Hooft line operators.

In the absence of matter in a charge $q$ representation of $O(2)$, the non-invertible, topological Gukov-Witten operator $T^{O(2)}(\theta)$ for $\theta \in (0, \pi)$ descends to a sum of invertible, topological Gukov-Witten operators $U^{U(1)}(\theta) +U^{U(1)}(-\theta) $ \eqref{O2GWsplitting}. For $\theta = \pi$, the invertible operator $T^{O(2)}(\pi)$ descends simply to $U^{U(1)}(\pi)$. These operators cease to be topological in the presence of charged matter, and the spectrum is complete precisely when no topological Gukov-Witten operators remain.

The non-endable Gukov-Witten operator of $O(2)$ gauge theory is no longer a genuine operator of the $U(1)$ gauge theory: instead, it represents the boundary of a codimension-1 surface operator. The det Wilson line, which was topological in $O(2)$ and linked with the non-endable Gukov-Witten operator in question, is trivial in the resulting $U(1)$ gauge theory. We see once again that topological operators in the UV theory descend to topological operators in the Higgsed theory, though they may be trivial in the latter theory.

\subsection{The Coulomb Branch of a General Gauge Theory}\label{ssec:Coulomb}

 We now consider the theory on the Coulomb branch of a general gauge group $G$, obtained by giving a vev to a scalar field $\Phi$ in the adjoint representation. 

By the structure theorem in Section \ref{ssec:Structure}, we may write $G$ as
\begin{equation}
G = \frac{G_0 \rtimes R}{P}\,,
\end{equation}
where $G_0$ is connected, $R$ is a finite group whose elements act on $G_0$ either trivially or via an outer automorphism of $G_0$, and $P$ is a finite, common subgroup of $Z(G_0)$ and $R$. For a generic vev of $\Phi$, this group will be broken to
\begin{equation}
H = \frac{U(1)^r \times K}{P}\,,
\end{equation}
where $r$ is the rank of $G_0$ and $K$ is the subgroup of $R$ that acts trivially on $G_0$. 

In pure $H$ gauge theory, every Gukov-Witten operator is topological. Such an operator is labeled by an element $g$ of $U(1)^r$ and a conjugacy class $[k]$ of $K$, modulo the subgroup $P$. After Higgsing from $G$, some of these Gukov-Witten operators are no longer topological, as some Wilson lines are endable. In particular (assuming that the only matter of the $G$ gauge theory is the adjoint-valued scalar $\Phi$ that acquires a vev), Wilson lines labeled by charges in the root lattice of $G_0$ are endable. Non-endable Wilson lines are therefore labeled by charges in the $U(1)^r$ charge lattice modulo the root lattice of $G_0$. The topological Gukov-Witten operators are the ones that link trivially with the endable Wilson lines, and they are associated with elements $(z, k)$, where $z \in Z(G_0)$ and $k \in K$.

At the origin of the Coulomb branch, where the full $G$ gauge symmetry is restored, Gukov-Witten operators are given by gauge-invariant sums of Gukov-Witten operators of the theory on the Coulomb branch. By continuity, the gauge-invariant sums of topological operators described above will remain topological at the origin of the Coulomb branch, so the Gukov-Witten operators associated with elements $(z, k)$, where $z \in Z(G_0)$ and $k \in K$ will be topological in the UV theory. As we argued in Section \ref{sec:generalTopGW}, these are precisely the Gukov-Witten operators that link trivially with the Wilson line in the adjoint representation. We see that these operators are indeed topological in $G$ gauge theory with an adjoint-valued scalar field, and thus they are topological in pure $G$ gauge theory as well.

\subsection{General Story for Higgsed Gauge Theories}

After discussing these particular examples, we provide the general picture. Consider a theory with compact gauge group $G$, which is Higgsed down to $H$ by a vev $\Phi$ in an arbitrary representation $\rho$. We will consider three kinds of operators: Wilson lines, Gukov-Witten operators, and 't Hooft operators. 

\subsubsection{Wilson Lines and Gukov-Witten Operators}

To begin, we note that the spectrum of the IR theory after Higgsing will be complete if and only if the spectrum of the UV theory is complete. First, assume that the UV theory is complete, so there exists matter in every representation of $G$. Under Higgsing, the representation $\rho$ of $G$ decomposes into a direct sum of Wilson lines of representations of $H$, as dictated by group theory ``branching rules''. Any representation of the IR gauge group $H$ necessarily shows up in the branching of some representation of the UV gauge group $G$ \cite{brocker2003representations}. This means that there will exist matter in every representation of $H$, and the IR gauge theory is also complete.

Conversely, suppose that the UV gauge theory is incomplete, so that at least one Wilson line $W_\nu$ is not endable. In the Higgsed theory, the Wilson line $W_\nu$ is identified, by inserting vevs of $\Phi$ or $\Phi^\dagger$, with Wilson lines in representations contained in tensor products of $\nu$ with tensor  powers of $\rho$ and ${\bar \rho}$. We claim that $W_\nu$ is never identified with an endable operator $W_\mu$. Consider the inner products of characters:
\begin{equation}
\langle \chi_\mu, \chi_{\nu \otimes \rho^n \otimes {\bar \rho}^m}\rangle = \langle \chi_\mu, \chi_\nu \chi_\rho^n \overline{\chi_\rho^m} \rangle = \langle \chi_\mu \chi_\rho^m \overline{\chi_\rho^n} , \chi_\nu\rangle = \langle \chi_{\mu \otimes \rho^m \otimes {\bar \rho}^n}, \chi_\nu \rangle. 
  \label{eq:characteridentities}
\end{equation}
If $\mu$ is contained in the decomposition of $\nu \otimes \rho^n \otimes {\bar \rho}^m$ into irreps, then $\langle \chi_\mu, \chi_{\nu \otimes \rho^n \otimes {\bar \rho}^m}\rangle \neq 0$, and we conclude from Eq.~\ref{eq:characteridentities} that $\nu$ is likewise contained in the decomposition of $\mu \otimes \rho^m \otimes {\bar \rho}^n$ into irreps. $W_\rho$ and $W_{\bar\rho}$ are endable in the UV theory, due to the existence of $\Phi$. Hence, if $W_\mu$ is endable, so is $W_\nu$, in contradiction with our initial assumption. Thus, the non-endable line $W_\nu$ descends to a nontrivial, non-endable line operator in the IR theory, which is incomplete as well.

We have argued in Section \ref{sec:General} that electric completeness is equivalent to the absence of topological Gukov-Witten operators. Since electric completeness is preserved under Higgsing, we learn that the existence of nontrivial topological operators is similarly preserved under Higgsing. 

The analogous statement does not hold for the topological Wilson lines of a theory. A topological Wilson in the UV theory will remain topological in the IR, but it may become trivial at the IR fixed point (as in the $q=2$ line of a $\mathbb{Z}_4$ gauge theory upon Higgsing to $\mathbb{Z}_2$, or the det line of $O(2)$ upon Higgsing to $U(1)$). When a topological line becomes trivial, the non-endable Gukov-Witten operators that link with it will cease to be genuine operators of the theory, and instead they will represent the boundaries of codimension-1 operators of the theory. However, a non-topological line operator of the UV theory will generically remain non-topological unless it becomes trivial. It may also happen that a non-topological line in the UV becomes (potentially) topological in the IR, such as a Wilson line of charge $q\not\equiv 0\,\text{mod}\, N$ in a $U(1)$ theory that is Higgsed to $\mathbb{Z}_N$. Such a line is never exactly topological, since at short enough length scales it behaves as the original line in the UV theory, and consequently the corresponding Gukov-Witten operators of the IR theory must be endable. We discuss how this happens below.

\subsubsection{'t Hooft Operators and ANO vortices}

Having dealt with Wilson lines and Gukov-Witten operators in the identity component, we now turn to 't Hooft operators. For simplicity, we work in four dimensions. 
 These can be discussed efficiently in terms of the long exact sequence in homotopy associated to the fibration:
 \begin{equation}
\begin{tikzcd}H\arrow{r}&G\arrow{r}&G/H,\end{tikzcd}\end{equation}
 which is
 
 \begin{equation}
 \begin{tikzcd} \cdots\arrow{r}&\pi_{2}(G)=0\arrow{r}&\pi_{2}(G/H)\arrow{dll} \\ \pi_{1}(H)\arrow{r} & \pi_{1}(G)\arrow{r}& \pi_{1}(G/H) \arrow{dll} \\ \pi_0(H) \ar{r} & \pi_0(G) \ar{r} & \pi_0(G/H) \end{tikzcd}\label{longef}\end{equation}
$\pi_2(G)$ vanishes for any Lie group. The second entry, $\pi_2(G/H)$, classifies 't Hooft/Polyakov monopole solutions. The fact that it injects into $\pi_1(H)$ (which classifies the 't Hooft operators of $H$, up to conjugation by $\pi_0(G)$) means that 't Hooft operators of $H$ in the image of the map are endable. The map from $S^1$ to $H$ becomes contractible when embedded in $G$, so one can deform the Higgs profile smoothly, similarly to the construction in Figure \ref{f3}. The resulting configuration is a pointlike operator from the point of view of the IR theory, and it creates an 't Hooft/Polyakov monopole. 

What about  't Hooft operators of $H$ that are not in the image of the map into $\pi_1(H)$? Exactness implies that they embed into $\pi_1(G)$---they descend from 't Hooft operators of the UV theory, in the usual fashion. These operators will be non-endable if and only if their associated UV line is non-endable. However, not every element of $\pi_1(G)$ is in the image of this map. Exactness tells us that elements of $\pi_1(G)$ that do not descend to 't Hooft operators of $H$ must map to classes in $\pi_1(G/H)$.  The field $\Phi$ winds around such an 't Hooft operator,  and there is an ANO vortex at the location where it vanishes. Indeed, $\pi_1(G/H)$ classifies (the nonabelian generalization of) ANO vortices \cite{Abrikosov:1956sx, Nielsen:1973cs,Alford:1992yx}.

This shows that the spectrum of non-endable 't Hooft operators matches in the UV and IR. UV 't Hooft operators either become attached to ANO vortices in the IR theory, or they descend to genuine 't Hooft operators. Additional lines may exist in the IR, but they can end on operators that create 't Hooft/Polyakov monopoles. 

If we continue to follow the sequence in \eqref{longef}, we learn about additional operators. Classes that are not in the image of the map from $\pi_1(G)$ to $\pi_1(G/H)$ will correspond to ANO vortices that cannot end on monopoles; they are stable, solitonic vortices of the IR theory. If $H$ is not connected, it will have codimension-2 operators labeled by an element of $\pi_0(H)$: these are Gukov-Witten operators associated to conjugacy classes not in the identity component of $H$. Gukov-Witten operators associated with classes in the image of the map from $\pi_0(G)$ into $\pi_0(H)$ can end on operators that create ANO vortices, and so become endable. These link with the new potentially topological Wilson lines in the IR theory alluded to in the previous Section. The rest of the Gukov-Witten operators associated with $\pi_0(H)$ descend from classes in $\pi_0(G)$, which label the corresponding Gukov-Witten operators in the UV theory. In this way, all Gukov-Witten operators in the IR descend from those of the UV, although some become endable. However, not every UV Gukov-Witten operator desdends to a Gukov-Witten operator of $H$; those in the image of the map into $\pi_0(G/H)$ force the Higgs field to jump around them. In the deep IR description, where the Higgs field is integrated out, this means that they are attached to codimension-1 operators, and thus become non-genuine  codimension-2 operators. The fact that there are dynamical domain walls for the Higgs field means that the corresponding codimension-1 operators are actually endable in the full theory. 

\subsubsection{Embedding in $SU(n)$}

To conclude this section, we will now show that any $H$ gauge theory can be obtained from Higgsing of an $SU(n)$ gauge theory, for $n$ sufficiently large. Since $SU(n)$ is simply connected, $\pi_1(SU(n))$ and $\pi_0(SU(n))$ are trivial, and therefore $SU(n)$ gauge theories have a complete spectrum of magnetic monopoles and twist vortices. Given an arbitrary compact Lie group $H$, we may therefore obtain $H$ gauge theory with a complete spectrum of magnetic monopoles and twist vortices by Higgsing from $SU(n)$ for an appropriate $n$. While we do not claim that this is always the method by which quantum gravity produces a complete spectrum of magnetic objects, it provides a proof of principle that it is always possible to add this collection of objects even at the level of a Lagrangian field theory.

In mathematical language, the unbroken subgroup $H$ of a symmetry group $G$ when a field $\Phi$ in representation $\mathbf{R}$ gets a vev $\Phi_0$ is called the ``stabilizer subgroup'' of $\Phi_0\in \mathbf{R}$. We now prove the following:

\begin{thm} Every compact Lie group $H$ arises as the stabilizer subgroup of some representation and vev of $SU(n)$, for $n$ sufficiently large. \end{thm}

\begin{proof}Any given compact Lie group $H$ has a faithful, unitary, finite-dimensional representation, and so embeds as a closed subgroup of $U(m)$ for some $m$ (see, e.g., \cite{Harlow:2018tng}, Theorem A.8; or Proposition 1 in Appendix I of Chapter IX of  \cite{bourbaki2008lie}). By composing with the inclusion $U(m) \hookrightarrow SU(m+1)$, we may thus realize $H$ as a closed subgroup of $SU(n)$ for $n = m+1$. Now, by Corollary 2 of the Equivariant Embedding Theorem (also known as the Mostow-Palais theorem) on page 374 of \cite{bourbaki2008lie}, any closed subgroup of a Lie group arises as the stabilizer subgroup for some element and representation. \end{proof} 

\section{Chern-Simons Terms}\label{sec:CS}

Throughout this paper, we have explored the relation between the absence of (non-invertible) global symmetries and completeness of the spectrum in the context of isolated  gauge groups. However, the story slightly changes when coupling the gauge field to some other $p$-form gauge field, e.g., via some BF coupling or Chern-Simons term. In this section, we will explain to what extent the conclusions drawn in this paper are still valid in these cases.

In short, we will see that the relation between endability of all extended operators and the absence of any topological operator still holds in the presence of Chern-Simons terms. However,  we will see that the two cases below imply some refinements of the Completeness Hypothesis. In the case of a BF coupling, the spectrum may be in some sense incomplete, even if every extended operator is endable. In the case of axion electrodynamics, we will see that one must consider all topological operators at once. The reason is the mixing between electric and magnetic symmetries induced by the higher-group structure \cite{Cordova:2018cvg, Hidaka:2020iaz, Hidaka:2020izy,Brennan:2020ehu} implied by the axion coupling. This is related to a refined notion of completeness, that includes both the types of particles as well as their worldvolume degrees of freedom.

\subsection{BF Theory}

Let us begin with BF theory in four dimensions, given by the following action,
\begin{equation}
S = \int \left( -\frac{1}{2 g^2} H_3 \wedge \star H_3 -\frac{1}{2 e^2} F_2 \wedge \star F_2+ \frac{m}{2 \pi}B_2 \wedge F_2 \right).
\end{equation}
Here, $m$ is an integer. The fate of the higher form global symmetries and their interpretation in terms of Chern-Weil symmetries was studied in detail in \cite{Heidenreich:2020pkc}. In the absence of charged states, for $|m| > 1$ the $U(1)$ electric 1-form is broken to a $\mathbb{Z}_m$ symmetry, and the same occurs for the 2-form global symmetry. The theory contains Wilson line and surface operators given by
\begin{alignat}{3}
&W_{n_1}(\gamma)&&=\exp\left(in_1\oint_\gamma A_1\right) \ ,\quad &&n_1=0,1,\dots m-1,\\
&W_{n_2}(\Sigma)&&=\exp\left(in_2\oint_\Sigma B_2\right) \ ,\quad &&n_2=0,1,\dots m-1.
\end{alignat}
The $A_1$ Wilson line is endable whenever $n_1$ is a multiple of $m$, since in that case it can end on an 't Hooft local operator of $B_2$ (namely, $\mathrm{e}^{i\phi}$, where $\phi$ is the dual axion). Similarly the $B_2$ Wilson surface of charge $n_2$ a mutiple of $m$ can end on an 't Hooft line of $A_1$. They also serve as symmetry operators of the electric symmetries, since
\begin{equation}
\begin{alignedat}{5}
&U_{\kappa_1}&&=\mathrm{e}^{i\kappa_1 \int \star H}&&=\mathrm{e}^{i\kappa_1 \int \rmd_A\phi}&&=\mathrm{e}^{i\kappa_1 \int A}&&=W_{n_1=\kappa_1},\\
&U_{\kappa_2}&&=\mathrm{e}^{i\kappa_2 \int \star F}&&=\mathrm{e}^{i\kappa_2 \int \rmd_B\tilde A}&&=\mathrm{e}^{i\kappa_2 \int B}&&=W_{n=\kappa}.
\end{alignedat}
\end{equation}
where we have used the equations of motion, $\rmd \star F=H=\star \rmd_A\phi$ and $\rmd\star H=F=\star \rmd_B \tilde A$ with $\tilde A$ being the dual magnetic gauge field.
In the deep IR, this BF theory is the $\mathbb{Z}_m$ gauge theory studied in Section~\ref{sec:ZN}, and indeed, unless the spectrum is complete and all Wilson lines and surfaces can end, there are topological symmetry operators.

By contrast, the magnetic side of the story is a bit more puzzling, especially when $|m|=1$ so there is no discrete global symmetry left. The 1-form magnetic global symmetry and the 0-form global symmetry are gauged, so it seems we get an incomplete magnetic spectrum (there are no monopoles nor instantons) even if there are no global symmetries. However, there are no genuine 't Hooft operators either (they are not gauge invariant), so it holds trivially that the absence of global magnetic symmetries is equivalent to the absence of non-endable 't Hooft operators. Whether UV-complete quantum gravity theories must satisfy a stronger, less trivial notion of completeness here is an interesting question, and we direct the reader to Section 3.2 of reference \cite{Heidenreich:2020pkc} for further discussion.

\subsection{Axion Electrodynamics}

Let us now consider $U(1)$ gauge theory in the presence of a $\theta$ angle:
\begin{equation}
S = \int \left( -\frac{1}{2 g^2} F \wedge \star F + \frac{\theta}{8 \pi^2}F \wedge F \right).
\end{equation}
Here, the transformation $\theta \rightarrow \theta+ 2\pi$ does not affect correlation functions of local operators at separated points. However, it does affect correlation functions of surface and line operators: in particular, under $\theta \rightarrow \theta + 2\pi$ a purely magnetic 't Hooft line of charge $m$ picks up electric charge $m$ via the Witten effect~\cite{Witten:1979ey}, and more generally the dyonic lines transform as $L_{n, m} \rightarrow L_{n+m k,m}$ under $\theta \rightarrow \theta + 2 \pi k$. Clearly, such a transformation will affect the spectrum of topological surfaces: in a theory with a magnetic monopole of charge $1$ but no electrically-charged matter, the Gukov-Witten operator $U_{g= \mathrm{e}^{i \alpha}}$ is topological for all $\alpha \in [0, 2\pi)$, and there is a $U(1)$ electric 1-form global symmetry. On the other hand, the magnetic 1-form symmetry is completely broken. After taking $\theta \rightarrow \theta+2\pi$, however, the monopole will become a dyon of charge $(1,1)$. Shifting the $\theta$ angle has produced a new theory, dual to the first, in which the magnetic 1-form symmetry is unbroken, but the dyonic 1-form symmetry generated by the operators $U_{g,g}$ is now broken. However, the anti-diagonal combination is preserved, and the surface operators $U_{g, g^{-1}}$ remain topological. More generally, although shifts of the $\theta$ angle may affect which line operators are endable and which surface operators are topological, they do not alter the 1-form global symmetry of the theory.

We now make $\theta$ dynamical by promoting it to an axionic field,
\begin{equation}
S = \int \left( -\frac{1}{2 g^2} F \wedge \star F -\frac12 f^2\, \rmd\phi\wedge \star \rmd\phi+ \frac{\phi}{8 \pi^2}F \wedge F \right),\label{act6}
\end{equation}
where now $\phi \to \phi + 2 \pi$ is a gauge symmetry.\footnote{With these conventions, the action \eqref{act6} is invariant under the axion shift symmetry only for manifolds where $\frac{1}{8\pi^2} \int F^2\in\mathbb{Z}$. This is the case for every spin 4-manifold.} 
Notice that the electric 1-form global symmetry is explicitly broken by the axion coupling, since
\beq
\rmd \star \rmd F_2 \propto \rmd\phi\wedge F_2\neq 0.
\eeq
Since the electric 1-form symmetry is broken, one might expect that the Wilson line operator of $A_1$ would be endable in axion electrodynamics. This is not the case, and indeed the Wilson line remains non-endable. To determine endability of some extended operator, it is enough to reason locally near a supposed endpoint of the operator, so we may assume we are working on $\mathbb{R}^4$. By the equation of motion of the gauge field, we can see that
\beq
{\rm d} \left(\frac{1}{g^2}\star F-\frac{1}{4\pi^2}\phi F\right) = 0.
\eeq
The term in parentheses is not gauge invariant under $\phi \to \phi + 2 \pi$, and so does not define a consistent current for a global 1-form symmetry. Nevertheless, the integral
\beq
\int_{S^2 \subset \mathbb{R}^4} \left(\frac{1}{g^2}\star F-\frac{1}{4\pi^2}\phi F\right),
\label{not_really_an_operator}
\eeq
is actually gauge invariant, where the integral is over a topologically trivial sphere $S^2 \subset \mathbb{R}^4$ linking our Wilson line. To see this, note that the flux of $F$ through a topologically trivial sphere vanishes. Thus, when $\phi \to \phi + 2 \pi$, the integral \eqref{not_really_an_operator} shifts by $\int_{S^2} \frac{1}{4 \pi} F$, which vanishes because this $S^2$ is topologically trivial.\footnote{Note that this argument breaks down if we were to either insert an 't Hooft line parallel to the Wilson line or add dynamical monopoles, since we would no longer be able to conclude that $\int_{S^2} F$ vanishes.}  Multiplying \eqref{not_really_an_operator} by $i$ and exponentiating, we obtain a topological surface operator that can only be defined on a topologically trivial $S^2$, but which links nontrivially with the Wilson line operators in $\mathbb{R}^4$. Though this is not a proper surface operator of the theory, it is enough to perform the argument illustrated in Figure \ref{fig:toplinking} and conclude that the Wilson lines cannot end on local operators in axion electrodynamics.\footnote{By contrast, the Wilson lines can end on an operator inserted at a conical singularity in spacetime, such as the cone on $S^1 \times S^2$. These operators are not local operators in QFT, but they do represent local operators in quantum gravity, and produce gravitational solitons carrying electric charge \cite{McNamara:2021cuo}.}

Hence, the Wilson lines remain non-endable, even though the axion coupling breaks the electric 1-form symmetry. This seems to contradict the statement that the absence of topological Gukov-Witten operators is in 1-to-1 correspondence with endable Wilson lines (i.e., a complete spectrum). But there is a milder statement that still holds true. Even if the electric Gukov-Witten operators are not topological, the theory contains other topological operators,
\beq
U_{\eta}(\mathcal{M}^2)=\exp\left(\frac{i\eta}{2\pi}\int_{\mathcal{M}^2}F\right)
\ ,\quad
U_{\alpha}(\mathcal{M}^1)=\exp\left(\frac{i\alpha}{2\pi}\int_{\mathcal{M}^1}\rmd\phi\right).
\label{other_axion_EM_symmetries}
\eeq
associated to the magnetic 1-form global symmetry and the 2-form global winding symmetry of $\phi$. Hence, we have not found a theory with non-endable Wilson lines but no topological operators at all.

The operators \eqref{other_axion_EM_symmetries} stop being topological if we introduce dynamical monopoles and strings. In fact, by the Witten effect, the presence of monopoles in axion electrodynamics also induces electric charge, so all Wilson lines and 't Hooft lines become endable. Similarly, as described in \cite{Callan:1984sa}, axion strings must also carry electric charges on their worldsheet, and a small loop of axion string with nonzero electric charge behaves as a dynamical electric particle as well. Hence, in the presence of monopoles and strings, all extended operators can end, in accordance with the fact that there are no topological operators left.

It is interesting to note that the naive 't Hooft line of $A_1$ and the winding surface of $\phi$ are non-genuine, since they are not invariant under gauge transformations of $\phi$ and $A_1$, respectively. Indeed, under $\phi \to \phi + 2 \pi$, the 't Hooft line transforms into a dyonic line, and so must be attached to a surface operator encoding this anomaly. Similarly, the winding surface of $\phi$ must be attached to a 3-volume operator defined by the exponentiated integral of the Chern-Simons 3-form of $A_1$. That these operators are non-genuine is a reflection of the anomaly inflow from the bulk to the worldvolume of probe monopoles or axion strings.

In order to define genuine operators representing the insertions of heavy probe monopoles or strings, we must include the worldvolume degrees of freedom that cancel the anomaly. We may do this by inserting the anomalous partition function of worldvolume fields alongside the naive 't Hooft line or winding surface operators,
\beq
L^\text{genuine}(\mathcal{M}^1) = L^{\text{naive}}(\mathcal{M}^1) Z_\text{w.v.}(\mathcal{M}^1), \quad S^\text{genuine}(\mathcal{M}^2) = S^{\text{naive}}(\mathcal{M}^2) Z_\text{w.v.}'(\mathcal{M}^2),
\eeq
For the 't Hooft line, we may take $Z_\text{w.v.}$ to be the partition function of a particle on a circle discussed in \cite{Cordova:2019jnf, Jackiw:1975ep, Heidenreich:2020pkc}. For the winding surface of the axion, we may take $Z_\text{w.v.}'$ to be a chiral boson of unit charge \cite{Callan:1984sa}, which sits at the boundary of a Chern insulator, described by the Chern-Simons partition function (see, e.g., \cite{Witten:2015aoa} for a review).

While the insertion of the partition function of an entire QFT along a submanifold may seem a bit exotic, it nevertheless defines a perfectly good extended operator in the bulk theory. Indeed, the picture we advocate here is quite general: just like their dynamical counterparts, probe particles and branes in general carry worldvolume degrees of freedom, which are described by a worldvolume QFT. Since the probes themselves have internal dynamics, the operators representing their insertion must include these degrees of freedom when they interact with the bulk. In cases like the above, where there is nontrivial anomaly inflow, the presence of these worldvolume degrees of freedom is essential in order to cancel the anomaly. Relatedly, while the probe brane inserted by the operator should be regarded as an infinitely massive object, it may not be possible to make the worldvolume degrees of freedom on it arbitrarily heavy, in particular if they are chiral and participate in an anomaly inflow mechanism. 

We can extract several lessons from this example. First, there is no longer a one-to-one correspondence between the absence of topological operators and the completeness of the spectrum. However, a more general correlation persists, in which the endability of every extended operator is still equivalent to the absence of any topological operator. This mixing between the symmetries reflects a higher-group structure \cite{Cordova:2018cvg, Hidaka:2020iaz, Hidaka:2020izy,Brennan:2020ehu}.

Secondly, in the presence of Chern-Simons terms, we do not need to introduce all charged states to make sure all extended objects can end, but rather only the subset that allows us to populate the sites of the charge lattice. For instance, by the Witten effect in axion electrodynamics, the presence of monopoles already implies the presence of states with electric charge. Moreover, the fact that the operators that insert probe monopoles or axion strings include the partition function of worldvolume degrees of freedom tells us that the proper notion of completeness in the presence of a higher-group structure includes the presence of charged states with appropriate worldvolume excitations.

\section{Discussion}\label{sec:CONC}

\subsection{Implications for the Swampland Program}

In this paper, we set out to understand the relationship between two basic Swampland conjectures, namely the absence of global symmetries and the completeness of the spectrum in consistent theories of quantum gravity. Though these are conjectures about the nature of quantum gravity, their relationship can be studied fully within the context of quantum field theory, and this is the approach we have taken. We have shown that while the absence of invertible symmetries is insufficient to imply completeness, the absence of more general, non-invertible symmetries is in fact sufficient, at least in the examples we have considered. We strongly suspect that this pattern is completely general, and that the endability of every extended operator follows from the absence of any topological operator in a general QFT (in three dimensional TQFTs, this follows from the modularity of the tensor category of line operators \cite{Rudelius:2020orz}).

A natural question, then, is whether these non-invertible topological operators are forbidden in consistent theories of quantum gravity. While the absence of invertible global symmetries has been established quite well \cite{banks:1988yz, kallosh:1995hi, Banks:2010zn, Harlow:2018tng, Harlow:2020bee, Chen:2020ojn,Belin:2020jxr, Yonekura:2020ino}, the standard arguments involving black hole physics or holography do not directly apply to rule out non-invertible topological operators \cite{Rudelius:2020orz}. Of course, if we assume the completeness hypothesis, then by the arguments outlined in this paper we would expect the absence of any topological operator, but this is begging the question, as the original goal was to argue for the completeness hypothesis in the first place. Thus, we would like an argument against general topological operators in quantum gravity that does not assume completeness.

In fact, the absence of topological operators is a specific case of the broader statement that quantum gravity should have no nontrivial bulk operators whatsoever. This is a form of background independence, since a nontrivial bulk operator represents a way to couple the bulk gravity theory to a non-dynamical, background probe. In this form, the absence of topological operators follows from the Baby Universe Hypothesis \cite{McNamara:2020uza}, since the insertion of any such operator would necessarily create a nontrivial baby universe state on the boundary of a tubular neighborhood of its support. A similar argument from the Baby Universe Hypothesis shows directly that any operator must be endable, as every bulk operator must be nonperturbatively gauge equivalent to the identity operator, which is clearly endable. These arguments may be less than convincing, as the Baby Universe Hypothesis is less well established than the Completeness Hypothesis itself, but it shows that the emerging picture in the Swampland program is internally consistent.

A further consideration is to ask whether our analysis, taking place purely within quantum field theory, may be altered by making gravity dynamical, even without assuming any Swampland conditions. In fact, work in progress~\cite{McNamara:2021cuo} suggests that this is indeed the case, and that there are new charged objects, given by gravitational solitons, whose charges must be taken into account when addressing the question of completeness. It would be very interesting to compare these two approaches, and understand the relationship between non-invertible symmetries in quantum field theory and the possible charges of gravitational solitons in general.

\subsection{Existence of Twist Strings}
 \label{sec:cosmicstrings}

A wide variety of proposed extensions of the Standard Model postulate that the laws of physics in our universe are invariant under local, discrete symmetries. For example, many models of dark matter postulate that it is charged under a discrete (finite) symmetry, which explains its stability. In quantum gravity, we expect that such discrete symmetries are either slightly broken (and hence, that rare symmetry violating processes can occur, such as dark matter decay) or gauged. Because a finite symmetry has no massless, dynamical propagating gauge boson associated with it, there is no easy way for a low-energy observer to conclude that such a symmetry is gauged. However, it has long been known that such symmetries can be associated with twist vortices---strings, in four dimensions---with the property that a charged particle that circles around the vortex comes back to itself only up to a gauge transformation~\cite{Krauss:1988zc}. Although it is a common viewpoint that the existence of such objects is a defining feature of a gauged discrete symmetry~\cite{Witten:2017hdv}, one might wonder whether they necessarily exist as dynamical objects in the theory. We have argued that the absence of topological Wilson lines---which generate a (possibly non-invertible) $(d-2)$-form global symmetry---is equivalent to the existence of a complete spectrum of such twist vortices. This provides a stronger argument that these objects should exist in theories of quantum gravity (at least in $d \geq 4$ dimensions, where we expect $p$-form symmetries are forbidden for all $p$). 

This argument suggests that a wide variety of cosmic strings could exist in our universe, with potentially observable consequences. For example, the proton could be stable if there is a discrete gauged $\mathbb{Z}_3$ or $\mathbb{Z}_6$ symmetry under which the proton is charged~\cite{Ibanez:1991hv,Ibanez:1991pr,Dreiner:2005rd}. Our arguments suggest that such theories should admit cosmic strings with Aharonov-Bohm interactions with baryons. When the discrete gauge group is not a Higgsed remnant of a compact connected gauge group, these cosmic strings may be fundamental objects, rather than semiclassical strings that can be derived within an effective field theory. Similar remarks would apply to dark matter stabilized by a discrete symmetry (e.g., a $\mathbb{Z}_2$ or $\mathbb{Z}_4$ $R$-symmetry in a supersymmetric context). 

The existence of twist strings has an important consequence for the spontaneous breaking of a discrete symmetry. When a global discrete symmetry is spontaneously broken, stable domain walls are produced~\cite{Zeldovich:1974uw, Kibble:1976sj}. Stable domain walls produced after inflation are a cosmological disaster: having an equation of state $w = -2/3$, they redshift more slowly than radiation or matter, and rapidly dominate the energy density of the universe. Explicit symmetry breaking can destabilize a domain wall. Gauging can {\em also} destabilize the domain wall, because domain walls can end on twist strings. This has two consequences. The first is that twist strings are confined to the ends of domain walls, and can experience a force that attracts them toward each other to annihilate. The second is that the nucleation of a loop of twist string inside a domain wall can create a hole in the wall, which can be eaten up by the appearance of such holes, although this process is highly suppressed if the tension of the twist string is too large~\cite{Kibble:1982dd}. Thus, the physics of twist strings potentially leads to major differences in the cosmology of models with gauged discrete symmetries and those with global discrete symmetries.

Gauged discrete symmetries that have nontrivial interplay with continuous gauge symmetries are also of interest. We have discussed the case of $O(2)$ gauge theory, in which the $\mathbb{Z}_2$ charge conjugation symmetry does not commute with the $U(1)$ electromagnetic symmetry, in some detail. This leads to the fascinating and well-studied physics of Alice strings~\cite{Schwarz:1982ec,Alford:1990mk,Preskill:1990bm}, which can convert electrons into positrons. Recently, a variety of extensions of the Standard Model have been based on the idea of discrete symmetries, like $\mathbb{Z}_N$, which act to {\em permute} copies of a gauge theory. For example, the Twin Higgs model postulates a mirror copy of the Standard Model, related by a $\mathbb{Z}_2$ exchange symmetry, which ameliorates quantum corrections to the Higgs boson mass~\cite{Chacko:2005pe}. The original incarnation of this model, and much of the subsequent work, have assumed a small explicit breaking of this symmetry, but variations of the model assume an exact $\mathbb{Z}_2$ symmetry that is spontaneously broken~\cite{Geller:2014kta,Beauchesne:2015lva}. We expect that such an exact symmetry would be gauged, and accompanied by cosmic twin strings, such that circulating (for instance) a Standard Model gluon around a twin string would convert it to a twin gluon. (A similar picture should apply to the two $E_8$ factors of the $E_8 \times E_8$ heterotic string, which are exchanged by a gauged $\mathbb{Z}_2$.) Successful Twin Higgs phenomenology requires that the $\mathbb{Z}_2$ symmetry be spontaneously broken, meaning that the twin strings would arise as boundaries of twin domain walls across which the roles of the Standard Model and its twin sector are exchanged. The destabilization of twin domain walls by twin strings could play an interesting role in the cosmological history of such models.

\subsection{Concluding Remarks}

The concept of symmetry has been a dominant one in the theoretical physics of the last hundred years. Recently, ever-more general symmetries, including $p$-form global symmetries and higher-group global symmetries, have come to play an increasingly central role in both condensed matter theory and high energy theory. Although the concept of a non-invertible symmetry is not new, here we have seen that it has a major role to play through its equivalence with an incomplete spectrum of charged objects in a theory.

Conjectures about universal properties of quantum gravity have proliferated in recent years, a phenomenon that has attracted some criticism from observers. Not all of these conjectures are independent, however, and we can make progress by paring them down to a smaller set of core ideas. Here we have shown that two of the oldest such conjectures---the absence of global symmetries and the Completeness Hypothesis---are both subsumed by the absence of topological operators in quantum gravity. That, in turn, is merely one facet of a much older observation that quantum gravity does not admit true local operators at all. There is more progress to be made in taming the growth of Swampland conjectures by untangling their connections to common roots. The simple, radical ideas about quantum gravity that unify the conjectures, and the ramifications of those ideas for particle physics and cosmology, will continue to bear fruit in the years to come.

\section*{Acknowledgments}

We thank Loren Spice for useful discussions of the structure of disconnected compact Lie groups on MathOverflow \cite{378141,378220,378257}. We thank Sergio Cecotti, Meng Cheng, Daniel Harlow, Po-Shen Hsin, Shu-Heng Shao, Ryan Thorngren, Cumrun Vafa, Qing-Rui Wang, Xiao-Gang Wen, and Xueda Wen for useful discussions. We also thank Shu-Heng Shao for comments on the manuscript. MR thanks Thomas Dumitrescu for explaining to him the connection between completeness and absence of generalized global symmetries in $U(1)$ theories, some years ago.
BH is supported by National Science Foundation grants PHY-1914934 and PHY-2112800.
JM is supported by the National Science Foundation Graduate Research Fellowship Program under Grant No.~DGE1745303.
The research of MM and IV was supported by a grant from the Simons Foundation (602883, CV). 
MR is supported in part by the DOE Grant DE-SC0013607, the NASA Grant 80NSSC20K0506, and the Alfred P.~Sloan Foundation Grant No.~G-2019-12504.
The work of TR at the Institute for Advanced Study was supported by the Roger Dashen Membership and by
NSF grant PHY-1911298. The work of TR at the University of California, Berkeley, was supported by NSF grant  PHY1820912, the Simons Foundation, and the Berkeley Center for Theoretical Physics.
We acknowledge hospitality from several institutions where portions of this work were completed: the Amherst Center for Fundamental Interactions at UMass Amherst, site of the workshop ``Theoretical Tests of the Swampland''; the Aspen Center for Physics, which is supported by National Science Foundation grant PHY-1607611; the 2019 Simons Summer Workshop, at the Simons Center for Geometry and Physics at Stony Brook University; and the KITP at UC Santa Barbara, supported in part by the National Science Foundation under Grant No.~NSF PHY-1748958.

\bibliography{ref-v2}
\bibliographystyle{utphys}
\end{document}